\documentclass[a4paper,11pt,fleqn]{article}
\usepackage[includeheadfoot,nofoot,rmargin=25mm,tmargin=20mm]{geometry}
\usepackage{fancyhdr,lastpage}
\usepackage[numbers]{natbib}
\usepackage{graphicx} 
\usepackage{amsmath} 
\usepackage{amssymb,amsthm}  
\usepackage{txfonts,sansmath}
\usepackage{url,color}
\usepackage{enumerate,paralist}
\usepackage[vlined,algoruled]{algorithm2e}
\usepackage{url,color}

\sansmath

\newcommand{\R}{\mathbb{R}}
\newcommand{\C}{\mathbb{C}}

\newcommand{\Z}{\mathcal{Z}}
\newcommand{\Rb}{\mathcal{R}}
\newcommand{\Sb}{\mathcal{S}}
\newcommand{\Tb}{\mathcal{T}}

\newcommand{\B}{\mathcal{B}}

\newcommand{\V}{\mathcal{V}}

\newcommand{\setU}{\mathcal{U}}

\newcommand{\img}{\mathrm{Im}\,}
\newcommand{\srad}{\rho}
\newcommand{\id}{\mathrm{I}}

\newcommand{\s}{\textsc{n}}
\newcommand{\CL}{\textsc{cl}}

\newcommand{\dfn}{:=}

\newcommand{\dist}{\mathrm{d}}
\newcommand{\md}{\underline{\dist}}
\newcommand{\vv}{\hat v^\natural}
\newcommand{\E}{E}
\renewcommand{\H}{H}
\newcommand{\cea}{\mathrm{CEA}}
\DeclareMathOperator{\argmin}{\mathrm{argmin}}
\newcommand{\Ind}{\underline{\s}}

\newtheorem{problem}{Problem}

\newtheorem{lemma}{Lemma}
\newtheorem{thm}{Theorem}

\newtheorem{defn}{Definition}
\newtheorem{rem}{Remark}

\newtheorem{claim}{Claim}

\title{\LARGE \bf Feedback stabilisation of switched systems via
  iterative approximate eigenvector assignment\thanks{Version from
    \today. Extended version of that submitted to CDC 2010.}}

\author{Hernan Haimovich%
  \thanks{H. Haimovich is with CONICET and the Laboratorio de Sistemas
    Din\'amicos y Procesamiento de Informaci\'on, Departamento de
    Control, Escuela de Ingenier\'{\i}a Electr\'onica, Facultad de
    Ciencias Exactas, Ingenier\'{\i}a y Agrimensura, Universidad
    Nacional de Rosario, Riobamba 245bis, 2000 Rosario, Argentina,
    {\tt\small haimo@fceia.unr.edu.ar}}\, %
  and %
  Julio H. Braslavsky
  \thanks{J.H. Braslavsky is with the ARC Centre for Complex Dynamic
    Systems and Control, The University of Newcastle, Callaghan NSW
    2308, Australia {\tt\small jhb@ieee.org}}%
}

\date{}

\begin{document}
\addtolength{\topmargin}{-0.5em}

\maketitle
\thispagestyle{empty}

\begin{abstract}
  This paper presents and implements an iterative feedback design
  algorithm for stabilisation of discrete-time switched systems under
  arbitrary switching regimes. The algorithm seeks state feedback
  gains so that the closed-loop switching system admits a common
  quadratic Lyapunov function (CQLF) and hence is uniformly globally
  exponentially stable. Although the feedback design problem
  considered can be solved directly via linear matrix inequalities
  (LMIs), direct application of LMIs for feedback design does not
  provide information on closed-loop system structure. In contrast,
  the feedback matrices computed by the proposed algorithm assign
  closed-loop structure approximating that required to satisfy
  Lie-algebraic conditions that guarantee existence of a CQLF. The
  main contribution of the paper is to provide, for single-input
  systems, a numerical implementation of the algorithm based on
  iterative approximate common eigenvector assignment, and to
  establish cases where such algorithm is guaranteed to succeed. We
  include pseudocode and a few numerical examples to illustrate
  advantages and limitations of the proposed technique.
\end{abstract}

\section{Introduction}

We consider the discrete-time switching system (DTSS)
\begin{eqnarray}
  \label{eq:dtss}
  x_{k+1}&=& A_{i(k)}x_k + B_{i(k)}u_k,
\end{eqnarray}
with $x_k\in\R^n$, $u_k\in\R^m$, defined by a switching function
\begin{equation*}
  i(k) \in \Ind \dfn \{1,2,\ldots,\s\},
  \quad \text{for all }k,
\end{equation*}
and a set of controllable subsystem pairs $\{(A_i,B_i):i\in\Ind\}$,
where the input matrices $B_i$, $i\in\Ind$, have full column rank.  

We address feedback control design of the form
\begin{equation}
  \label{eq:8}
  u_k = K_{i(k)} x_k,
\end{equation}
(which assumes that at every time instant $k$ the ``active'' subsystem
given by $i(k)$ is known) so that the resulting closed-loop system
\begin{align}
  \label{eq:26}
  x_{k+1} &=  A_{i(k)}^\CL x_k,\quad\text{where}\\
  \label{eq:49}
  A_{i}^\CL &= A_{i} + B_{i} K_{i},\quad\text{for }i\in\Ind,
\end{align}
be exponentially stable under arbitrary switching.  

It is well-known that ensuring that the closed-loop matrices $A_i^\CL$
are stable for each $i\in\Ind$ is necessary but not sufficient to
ensure the stability of the DTSS (\ref{eq:26})--(\ref{eq:49}) under
arbitrary switching \cite{liberzon99:_basic}.  A necessary and
sufficient condition for uniform exponential stability under arbitrary
switching is the existence of a common Lyapunov function for each of
the component subsystems in (\ref{eq:26})--(\ref{eq:49})
\cite{molchanov89:_criter}.  Such Lyapunov functions, however, will in
general have complex level sets, which makes their numerical search
difficult \cite{shorten07:_stabil}.

The search for common \emph{quadratic} Lyapunov functions (CQLF),
although restrictive, is appealing, since these functions play an
important role in the stabilisation of linear time-invariant systems
such as the component subsystems in (\ref{eq:dtss}).  The design of
feedback matrices $K_i$ in (\ref{eq:8}) so that the DTSS
(\ref{eq:26})--(\ref{eq:49}) admits a CQLF may be posed as follows.
\begin{problem}
  \label{prb:main}
  Given the matrices $A_i\in\R^{n\times n}$ and $B_i\in\R^{n\times m}$
  for $i\in\Ind$, design feedback matrices $K_i\in\R^{m\times n}$ such
  that the DTSS closed-loop system (\ref{eq:26})--(\ref{eq:49}) admits
  a CQLF.
\end{problem}

Quadratic Lyapunov functions are amenable to efficient numerical
computation using linear matrix inequalities (LMIs).  For example,
Problem~\ref{prb:main} can be solved by finding $X=X^T > 0$ and $N_i$
to satisfy the LMIs
\begin{equation}
  \label{eq:13}
  \begin{bmatrix}
    X   &(A_iX+B_iN_i)^T \\
    A_iX + B_i N_i & X
  \end{bmatrix} > 0,\quad i\in\Ind,
\end{equation}
where the required feedback matrices are given by $K_i = N_i X^{-1}$
and the CQLF is $V(x)=x^T X^{-1} x$ (see for example
\cite{daafouz02:_stabil}, \cite{sala05:_comput}).  An advantage of
this approach is that the feasibility of the LMIs (\ref{eq:13}) is
necessary and sufficient for the DTSS considered to admit a CQLF.
However, blind application of such control design strategy gives no
insight on the structure of the closed-loop DTSS. Thus, as pointed out
in \cite{wulff09:_contr_dsgn}, these LMI methods ``lack transparency
and interpretability that was a feature of classical techniques'' and
hence ``a pressing need remains for analytic tools to support the
design of stable switched systems'' (also see
\cite{shorten07:_stabil}, \cite{lin09:_stabil} for similar comments).

As an alternative to the LMI approaches, the authors in
\cite{wulff09:_contr_dsgn} propose a pole-placement technique for
single-input single-output continuous-time switching systems.  The
strategy in \cite{wulff09:_contr_dsgn} seeks to guarantee closed-loop
uniform global exponential stability under arbitrary switching by
designing controllers that achieve a closed-loop common eigenvector
structure.  By constraining such eigenstructure and the class of
controllers allowed, the strategy in \cite{wulff09:_contr_dsgn}
simplifies the design process, providing analytically transparent
solutions in a restricted but practically important class of systems.

The present paper presents another strategy that seeks to ``activate''
analytic tools into a feedback design methodology (much in the spirit
of \cite{kokarc_aut01}) to solve Problem~\ref{prb:main}.  Our strategy
follows from the previous paper \cite{HBF2009:_CDC}, which introduced
an iterative algorithm to seek feedback gains that make the set of
closed-loop subsystems (\ref{eq:26})--(\ref{eq:49}) satisfy Lie
algebraic conditions that guarantee the existence of a CQLF
\cite{liberzon03:_switc}.  While such Lie-algebraic conditions are
restrictive, since they are not necessary for the existence of a CQLF,
we believe they offer an insightful way to understand and exploit
fundamental system structure in feedback design for DTSS.

The theoretical results in \cite{HBF2009:_CDC} brought forward the
following important consequences: (i) if the proposed Lie-algebraic
feedback design problem is feasible, its solution can be found in an
iterative manner (similarly to the way solvability of Lie-algebras can
be checked for autonomous switched systems) by seeking feedback gains
that assign a single common stable eigenvector at each iteration step;
and (ii) in seeking such common eigenvector, if at any iteration step
more than one vector can be assigned by feedback, it is irrelevant
which one is chosen by the procedure.  These observations provide
motivation to the development of numerical implementations, which were
not discussed in \cite{HBF2009:_CDC}.

The present paper focuses on the numerical implementation aspects of
the iterative algorithm introduced in \cite{HBF2009:_CDC}.  A key
question in the proposed approach is the lack of robustness of the
Lie-algebraic conditions that are sought to satisfy by feedback
design.  Indeed, it is well-known that these conditions are destroyed
by arbitrarily small perturbations to the individual matrices
\cite[\S2.2.4]{liberzon03:_switc}, and thus, there are a priori no
guarantees that the algorithm in \cite{HBF2009:_CDC} can find any
solution at all in a (necessarily approximate) numerical
implementation.

However, suppose that for a given set of systems to be stabilised
there exists, in a neighbourhood of the original, a feasible (a priori
unknown) set of systems, which is stabilisable and such that the
resulting closed-loop systems satisfy the Lie-algebraic conditions
that guarantee the existence of a CQLF.  Then, by continuity of such
CQLF, and if the neighbourhood is sufficiently small, there will also
exist a CQLF for the original set of systems, despite the fact that it
may be impossible to make them satisfy the Lie-algebraic conditions.

The main contribution of the present paper is to, based on the above
argument, derive and mathematically justify a specific numerical
implementation, for single-input systems, of the design algorithm
proposed in \cite{HBF2009:_CDC} that will succeed not only in cases
where the Lie-algebraic conditions are satisfied but also in
approximate cases. The proposed numerical implementation is based on
the solution of an optimisation problem that seeks feedback matrices
that will achieve closed-loop systems with an \emph{approximate}
common eigenvector. The existence of a CQLF for the corresponding
closed-loop DTSS may be readily checked a posteriori with a set of
\emph{informed} LMIs. This step is necessary, since it is in general a
priori unknown if the ``exact'' problem is feasible in some
neighbourhood of the original system data. The algorithm has been
implemented in \textsc{Matlab}$^{\tiny\textregistered}$, and numerical
examples to illustrate its application and discuss its advantages and
limitations are provided.

The rest of the paper is organised as follows.  The proposed algorithm
and its core, a procedure for approximate common eigenvector
assignment (Procedure~CEA), are presented and explained in
Section~\ref{sec:feedb-contr-design}.  The main theoretical results,
justifying the numerical implementation of the proposed algorithm for
single-input systems, appear in Section~\ref{sec:enabling-results}.
Section~\ref{sec:examples} presents details about the
\textsc{Matlab}$^{\tiny\textregistered}$ implementation of the
algorithm and some illustrative numerical examples, and
Section~\ref{sec:conclusions} gives the paper conclusions.  The main
proofs (Theorems~\ref{thm:num1} and~\ref{thm:KiCQLF}) are given in the
Appendix.

\paragraph*{Notation}

$\R$ and $\C$ denote the real and complex numbers. $\|\cdot\|$ denotes
Euclidean norm or corresponding induced matrix norm. $M^*$ denotes the
conjugate transpose of $M$. $\srad(\cdot)$ denotes spectral radius and
$(M)_{:,k}$ is the $k$-th column of $M$. If $M\in\C^{n\times m}$,
$\img M$ denotes $\{x\in\C^n : x = M y, y \in\C^m\}$. The Euclidean
distance between a vector $v\in\C^n$ and a set $\V\subset\C^n$ is
denoted $\dist(v,\V)$. $j$ denotes $\sqrt{-1}$.


\tableofcontents

\section{Feedback Control Design}
\label{sec:feedb-contr-design}

A sufficient condition for the closed-loop DTSS
(\ref{eq:26})--(\ref{eq:49}) to admit a CQLF is given by the following
result, which is a minor modification of \cite[Theorem
6.18]{theys_phd05}.
\begin{lemma}
  \label{lem:stab}
  If $\srad(A_i^\CL)<1$ for $i\in\Ind$, and the Lie algebra
  generated by $\{A_i^\CL: i\in\Ind\}$ is solvable, then
  (\ref{eq:26})--(\ref{eq:49}) admits a CQLF.
\end{lemma}
In matrix terms, the fact that the Lie algebra generated by
$\{A_i^\CL: i\in\Ind\}$ be solvable is equivalent to the
existence of a single invertible matrix $T\in\C^{n\times n}$ such that
$T^{-1} A_i^\CL T$ is upper triangular for $i\in\Ind$ (even if
$A_i^\CL$ have real entries, those of $T$ may be complex
\cite{erdwil_book06}). 

\subsection{The algorithm}
\label{sec:algorithm}

In \cite{HBF2009:_CDC}, we established that given $A_i$ and $B_i$,
there exist feedback matrices $K_i$ that cause the Lie algebra
generated by $A_i^\CL$ to be solvable if and only if such feedback
matrices can be computed by an algorithm that performs iterative
common eigenvector assignment by feedback. A matrix version of such
algorithm is given in pseudocode as Algorithm~\ref{alg:main}.
\begin{algorithm}[!ht]
  \label{alg:main}
  \SetTitleSty{texrm}{\normalsize}
  \SetKwComment{tcc}{\% }{}
  \KwData{$A_i \in \R^{n\times n}$, $B_i \in \R^{n\times 1}$ for
    $i\in\Ind$, $\epsilon_c > 0$, $\epsilon_d > 0$}
  \KwOut{$U$, $K_i$ for $i\in\Ind$}
  \SetKwBlock{Init}{begin}{end}
  \Init(Initialisation){
    $A_i^1 \dfn A_i$, $B_i^1 \dfn B_i$, $K_i^0 \dfn 0$, $U_1 \dfn
    \id$ \;
    $U \dfn [\,]$ (empty), $\ell \leftarrow 0$ \;
  }
  \Repeat{$\ell = n$}%
  {\vspace{-1em}
    \begin{align}
      \notag
      \ell &\leftarrow \ell+1,\quad n_r \leftarrow n-\ell+1 \;\\
      \label{eq:15}
      [v_1^\ell,\{F_i^\ell\}_{i=1}^{\s}] &\leftarrow  \cea(\{A_i^\ell,
      B_i^\ell\}_{i=1}^{\s}, \epsilon_c, \epsilon_d)\;
    \end{align}
    Define
    \begin{align}
      \label{eq:33}
      A_i^{\ell,\CL} &\dfn A_i^\ell + B_i^\ell F_i^\ell,\\
      \notag
      (U)_{:,\ell} &\leftarrow \biggl(\prod_{r=\ell}^{r=1} U_r\biggr)
      v_1^\ell = U_{1} U_{2} \cdots U_{\ell} v_1^\ell,\\
      \label{eq:9}
      K_i^\ell &\leftarrow K_i^{\ell-1} + F_i^\ell 
      \biggl(\prod_{r=1}^{\ell} U_r^* \biggr)
    \end{align}
    \If{$\ell<n$}{%
      Construct unitary matrix with first column $v_1^\ell$:
      \begin{equation}
        \label{eq:62}
        \big[v_1^\ell | v_2^\ell | \cdots 
        | v_{n_r}^\ell\big] \in \C^{n_r \times n_r}.
      \end{equation}
      Assign
      \begin{align}
        \label{eq:99}
        U_{\ell+1} &\leftarrow [v_2^{\ell}|\cdots |v_{n_r}^\ell], \\
        \label{eq:100}
        A_i^{\ell+1} &\leftarrow U_{\ell+1}^{*} A_i^{\ell,\CL} U_{\ell+1},\\
        \label{eq:6} 
        B_i^{\ell+1} &\leftarrow U_{\ell+1}^{*} B_i^\ell,
      \end{align}
    }
  }
  $K_i \leftarrow K_i^n$ \;
  \caption{Approximate triangularisation by feedback}
\end{algorithm}

At every iteration, Algorithm~\ref{alg:main} executes Procedure~$\cea$
[see (\ref{eq:15})]. This procedure attempts to find feedback matrices
$F_i^\ell$ and a vector $v_1^\ell$ so that $(A_i^\ell + B_i^\ell
F_i^\ell)v_1^\ell = \lambda_i^\ell v_1^\ell$ and $|\lambda_i^\ell|<1$
for $i\in\Ind$, i.e. so that $v_1^\ell$ becomes a common eigenvector
of a set of closed-loop matrices, with corresponding stable
eigenvalues. The parameters $\epsilon_c$ and $\epsilon_d$ given as
arguments to  Procedure~$\cea$ are required for numerical reasons,
and will be explained in Section~\ref{sec:comm-eigenv-assignm}.

\begin{rem}
  \label{rem:ut}
  It is straightforward to check that if Algorithm~\ref{alg:main}
  terminates successfully and at every iteration ($\ell=1,\ldots,n$)
  Procedure~$\cea$ is able to find the vector $v_1^\ell$ and feedback
  matrices $F_i^\ell$ such that $(A_i^\ell + B_i^\ell
  F_i^\ell)v_1^\ell = \lambda_i^\ell v_1^\ell$ for $i\in\Ind$,
  then the matrices $A_i^\CL$ given by (\ref{eq:49}) with $K_i$ as
  computed by the algorithm are such that $U^* A_i^\CL U$ are upper
  triangular and $\lambda_i^\ell$ is the $\ell$-th main-diagonal entry
  of $U^* A_i^\CL U$.
\end{rem}

  Note that a slight modification of Algorithm~\ref{alg:main} is
  necessary to ensure that real feedback matrices $K_i$ are
  computed. We do not explain such modification here due to space
  limitations, and because it does not add any essential information
  to our main results. The implemented computational version of the
  algorithm, employed in Section~\ref{sec:examples}, does indeed
  ensure such condition.

Our result in \cite{HBF2009:_CDC} established that a vector $v_1^\ell$
and feedback matrices $F_i^\ell$ such that $(A_i^\ell + B_i^\ell
F_i^\ell)v_1^\ell = \lambda_i^\ell v_1^\ell$ with $|\lambda_i^\ell| <
1$, for $i\in\Ind$, will exist at every iteration of
Algorithm~\ref{alg:main} if and only if there exist $K_i$ such that
$A_i^\CL$ as in (\ref{eq:49}) generate a solvable Lie algebra and
satisfy $\srad(A_i^\CL)<1$. Such result was of a theoretical nature,
since in a numerical implementation determining whether a vector is
exactly an eigenvector or is close to being so is an extremely
difficult task \cite{mormor_cdc97}.


\subsection{Approximate common eigenvector assignment}
\label{sec:comm-eigenv-assignm}

In this section, we provide a numerical implementation of
Procedure~$\cea$ executed by Algorithm~\ref{alg:main} and establish
some of its properties.
We give Procedure~$\cea$ in pseudocode first, and next define and
explain each of its parts.
\begin{procedure}[!ht]
  \SetTitleSty{texrm}{\normalsize}
  \SetKwComment{tcc}{\%}{}
  \KwIn{$A_i^\ell \in \C^{n_r \times n_r}$, $B_i^\ell \in \C^{n_r
      \times 1}$, for $i\in\Ind$, $\epsilon_c$, $\epsilon_d$}
  \KwOut{$v_1^\ell$, $F_i^\ell$ for $i\in\Ind$}
  \eIf{%
    $n_r=1$
  }%
  {%
    Select $A_i^{\ell,\CL}$ such that $|A_i^{\ell,\CL}|\le 1-\epsilon_c$\;
    $v_1^\ell \leftarrow 1$,\quad $F_i^\ell \leftarrow -(B_i^\ell)^{-1} A_i^\ell +
    A_i^{\ell,\CL}$ \;
  }%
  {%
    \eIf{%
      $\Sb(\epsilon_c,\epsilon_d) = \emptyset$
    }%
    {%
      \textbf{Stop}: unsuccessful termination.
    }%
    {%
      $v_1^\ell \leftarrow \argmin_{v\in\Sb(\epsilon_c,\epsilon_d)} J(v)$ \;
      $F_i^\ell \leftarrow M_i(v_1^\ell)$ \;
    }%
  }
  \caption{CEA()}
  \label{proc:cea}
\end{procedure}

Note that $n_r$ is the state dimension, which decreases by 1 at every
iteration of Algorithm~\ref{alg:main}. Hence, the case $n_r=1$ in
Procedure~$\cea$ corresponds to the trivial case of a one-dimensional
single-input system. If every subsystem is controllable, then $0 \neq
B_i^\ell \in \C$ and hence $A_i^{\ell,\CL}$ can be arbitrarily chosen.

For the case $n_r > 1$, Procedure~$\cea$ utilises the matrices
\begin{align}
  \label{eq:17}
  \E_i(v) &\dfn (vv^* - \id)A_i^\ell,\\
  \label{eq:18}
  \H_i(v) &\dfn (vv^* - \id)B_i^\ell,\\
  \label{eq:35}
  M_i(v) &\dfn -(\H_i(v)^* \H_i(v))^{-1} \H_i(v)^* \E_i(v),\\
  \label{eq:30}
  A_i^{\ell,\CL}(v) &\dfn A_i^\ell + B_i^\ell M_i(v),\\
  \intertext{the sets}
  \label{eq:60}
  \Sb_1 &\dfn \{v \in \C^{n_r} : \| v \| = 1 \}\\
  \label{eq:61}
  \Sb_2(\epsilon_c) &\dfn \bigcap_{i=1}^{\s} \Big\{v \in \C^{n_r} : \|
  A_i^{\ell,\CL}(v) v \| \le 1 - \epsilon_c \Big\},\\
  \label{eq:65}
  \Sb_3(\epsilon_d) &\dfn \bigcap_{i=1}^{\s} \Big\{v \in \C^{n_r} :
  \dist(v,\img B_i^\ell)\ge \epsilon_d \Big\},\\
  \label{eq:29}
  \Sb(\epsilon_c,\epsilon_d) &\dfn \Sb_1 \cap \Sb_2(\epsilon_c) \cap \Sb_3(\epsilon_d),\\
  \intertext{and the cost function}
  \label{eq:36}
  J(v) &\dfn \sum_{i=1}^{\s} \Big\| \big[\E_i(v) + \H_i(v) M_i(v) \big]  v \Big\|^2.
\end{align}

If $n_r>1$, Procedure~$\cea$ checks whether the set
$\Sb(\epsilon_c,\epsilon_d)$ given by (\ref{eq:29}) is empty. Comments
on the case $\Sb(\epsilon_c,\epsilon_d)=\emptyset$ are given later in
Remark~\ref{rem:infeas}. If $\Sb(\epsilon_c,\epsilon_d)\neq\emptyset$,
Procedure~$\cea$ searches for a vector $v_1^\ell$ that minimises
$J(v)$ as given by (\ref{eq:36}), subject to the constraint
$v\in\Sb(\epsilon_c,\epsilon_d)$. The constraint set
$\Sb(\epsilon_c,\epsilon_d)$ in (\ref{eq:29}) is the intersection of
three sets: $\Sb_1$ in (\ref{eq:60}), which constrains the search to
unit vectors; $\Sb_2(\epsilon_c)$ in (\ref{eq:61}), which imposes the
stability constraint $|\lambda_i^\ell|<1$, as we will shortly
demonstrate; and $\Sb_3(\epsilon_d)$ in (\ref{eq:65}), which is
included for technical reasons discussed next and justified in the
next section.

Note that if $\epsilon_d>0$ and $v\in\Sb_3(\epsilon_d)$, then
(\ref{eq:65}) implies that $v \notin \img B_i^\ell$, for
$i\in\Ind$. It follows that $\H_i(v)$ has the same (column) rank as
$B_i^\ell$ [see (\ref{eq:18})] and, in this case, $M_i(v)$ in
(\ref{eq:35}) is well-defined if $B_i^\ell$ has full column rank
(l.i. columns).

Thus, Procedure $\cea$ requires $B_i^\ell$ to have full column
rank. At the first iteration of Algorithm~\ref{alg:main}, we have
$B_i^1 = B_i$ and hence $B_i^1$ has full column rank by assumption.
At subsequent iterations of the algorithm, such condition is ensured
by the following result.
\begin{lemma}
  \label{lem:vnotimBfr}
  Let $B_i^\ell$ have full column rank and let $v_1^\ell$ be computed
  by Procedure~$\cea$. Then, $B_i^{\ell+1}$ has full column rank.
\end{lemma}
\begin{proof}
  Since $v_1^\ell$ is effectively computed by Procedure~$\cea$ and
  since $\Sb(\epsilon_c,\epsilon_d)$ is topologically closed, then
  $v_1^\ell \in \Sb(\epsilon_c,\epsilon_d)$ and hence $v_1^\ell \notin
  \img B_i^\ell$. The result then follows straightforwardly from the
  fact that the columns of (\ref{eq:62}) form a basis, (\ref{eq:99})
  and (\ref{eq:6}).
\end{proof}

We are now ready to show in what sense the minimisation of $J(v)$ is
related to the assignment of a common eigenvector.
\begin{lemma}
  \label{lem:comeig}
  Let $\epsilon_c>0$, $\epsilon_d>0$ and
  $v\in\Sb(\epsilon_c,\epsilon_d)$. Then,
  \begin{enumerate}[i)]
  \item $J(v) \ge 0$,\label{item:1}
  \item $J(v)=0$ if and only if $A_i^{\ell,\CL}(v)v = \lambda_i^\ell v$
    with $|\lambda_i^\ell|\le 1-\epsilon_c$.\label{item:2}
  \item There exists $G_i\in\C^{m\times n_r}$ such that
    \begin{equation}
      \label{eq:78}
      (A_i^\ell + B_i^\ell G_i)v = \lambda_i^\ell v,
    \end{equation}
    if and only if $A_i^{\ell,\CL}(v) v = \lambda_i^\ell v$.\label{item:3}
  \end{enumerate}
\end{lemma}
\begin{proof}
  \ref{item:1}) Straightforward from (\ref{eq:36}) and since $J$ is
  well defined on $\Sb(\epsilon_c,\epsilon_d)$.

  \ref{item:2}) ($\Rightarrow$) $J(v)=0$ implies that
  \begin{align}
    0 &= \Big\| \big[ \E_i(v) + \H_i(v) M_i(v) \big] v \Big\| \notag\\
    \label{eq:92}
    &= \| (vv^*-\id) [A_i^\ell + B_i^\ell M_i(v)] v \|,
  \end{align}
  for $i\in\Ind$, where we have used (\ref{eq:17}), (\ref{eq:18})
  and (\ref{eq:36}). In turn, (\ref{eq:92}) implies that
  \begin{equation}
    \label{eq:93}
    [A_i^\ell + B_i^\ell M_i(v)] v = A_i^{\ell,\CL}(v) v = \lambda_i^\ell v,
  \end{equation}
  where we have used (\ref{eq:30}). From (\ref{eq:93}),
  $|\lambda_i^\ell| = \| A_i^{\ell,\CL}(v) v \| \le 1 - \epsilon_c <
  1$ since $v\in\Sb(\epsilon_c,\epsilon_d)$.

  ($\Leftarrow$) Left-multiply (\ref{eq:93}) by $(vv^* - \id)$ to
  obtain (\ref{eq:92}), which implies $J(v)=0$.

  \ref{item:3}) ($\Rightarrow$) Left-multiplying (\ref{eq:78}) by
  $(vv^* - \id)$ yields $[\E_i(v) + \H_i(v) G_i]v = 0$. Let $u_i \dfn
  G_i v$.  Then, $\E_i(v)v = -\H_i(v) u_i$. Since $\H_i(v)$ has full
  column rank, then $u_i = M_i(v) v = G_i v$ and hence by
  (\ref{eq:78}) $A_i^{\ell,\CL}(v) v = [A_i^\ell + B_i^\ell M_i(v)] v
  = \lambda_i^\ell v$.

  ($\Leftarrow$) Just take $G_i = M_i(v)$.
\end{proof}
Lemma~\ref{lem:comeig} shows that, for
$v\in\Sb(\epsilon_c,\epsilon_d)$, $J(v)=0$ if and only if $v$ can be
assigned by feedback as a common eigenvector with corresponding stable
eigenvalues, and $J(v) > 0$ otherwise. Therefore, the search for a
vector that minimises $J(v)$ can be interpreted as the search for a
vector that is closest to an assignable common eigenvector. 
\begin{rem}
  \label{rem:infeas}
  If $\Sb(\epsilon_c,\epsilon_d) = \emptyset$, then no feasible vector
  exists for the optimisation performed by
  Procedure~$\cea$. Therefore, no vector can be computed, and
  Procedure~$\cea$ and hence Algorithm~\ref{alg:main} will terminate
  unsuccessfully. In the next section, we show that there exist
  $\epsilon_c,\epsilon_d>0$ so that $\Sb(\epsilon_c,\epsilon_d) \neq
  \emptyset$ in the cases of interest.
\end{rem}

In the next section, we establish conditions under which the vector
and feedback matrices computed by Procedure~$\cea$ cause
Algorithm~\ref{alg:main} to yield feedback matrices so that the
closed-loop DTSS admits a CQLF.

\section{Main Results}
\label{sec:enabling-results}

We now derive several results that justify the use of
Algorithm~\ref{alg:main} for feedback control design of single-input
systems.  Recall that the constraint set $\Sb(\epsilon_c,\epsilon_d)$
in the optimisation solved in Procedure~$\cea$ forces the search to be
performed over vectors $v$ that satisfy $v\notin\img B_i^\ell$. The
following result justifies this constraint for single-input DTSS with
controllable subsystems.
\begin{lemma}
  \label{lem:SIgood}
  Let $n_r>1$, $A_i^\ell\in\C^{n_r\times n_r}$,
  $B_i^\ell\in\C^{n_r\times 1}$, $(A_i^\ell,B_i^\ell)$ be controllable
  and suppose that $v\neq 0$ and $F_i^\ell$ satisfy
  $(A_i^\ell+B_i^\ell F_i^\ell)v = \lambda_i v$, for
  $i\in\Ind$. Then $v\notin\img B_i^\ell$ for $i\in\Ind$.
\end{lemma}
\begin{proof}
  Suppose for a contradiction that $v\in\img B_k^\ell$ for some
  $k\in\Ind$. Then $v=B_k^\ell u$ for some $0\neq u\in\C$. We
  have $(A_k^\ell+B_k^\ell F_k^\ell)B_k^\ell u = \lambda_k B_k^\ell u$
  and hence $B_k^\ell$ and $A_k^\ell B_k^\ell$ are linearly dependent,
  which prevents the pair $(A_k^\ell,B_k^\ell)$ from being
  controllable because $n_r>1$.
\end{proof}
Lemma~\ref{lem:SIgood} establishes that, for a single-input DTSS with
controllable subsystems, every vector that can be made a common
eigenvector by feedback will not be contained in the image of any
input matrix $B_i^\ell$. Therefore, in this case such constraint on
the optimisation problem is justified.

In the sequel, we will say that a set of matrix pairs is controllable
if every matrix pair in the set is controllable. We also require the
following definitions.
\begin{defn}[CEAS, $\gamma$-CEAS]
  \label{def:CEAS}
  A set of matrix pairs $\Z^\ell = \{(A_i^\ell\in\C^{n_r\times n_r},
  B_i^\ell\in\C^{n_r\times 1}) : i\in\Ind\}$ is said to be \emph{CEAS
    (Common Eigenvector Assignable with Stability)} if there exist
  $0\neq v\in\C^{n_r}$, $\lambda_i^\ell\in\C$ with $|\lambda_i^\ell|<1$, and
  matrices $F_i^\ell\in\C^{1\times n_r}$, such that
  \begin{equation}
    \label{eq:39}
    (A_i^\ell +B_i^\ell F_i^\ell)v = \lambda_i^\ell v,\quad\text{for }i\in\Ind,
  \end{equation}
  If $\Z^\ell$ is CEAS, we say that $v\in\C^{n_r}$ and $F_i\in\C^{1\times
    n_r}$ are \emph{compatible with $\Z^\ell$} whenever (\ref{eq:39})
  holds for $|\lambda_i^\ell|<1$. If (\ref{eq:39}) holds with
  $|\lambda_i|\le 1-\gamma$ for some $0< \gamma\le 1$, we say that
  $\Z^\ell$ is $\gamma$-CEAS and refer to the corresponding $v$ and
  $F_i^\ell$ as $\gamma$-compatible with $\Z^\ell$.
\end{defn}
\begin{defn}[SLASF, $\gamma$-SLASF]
  \label{def:SLASF}
  A set of matrix pairs $\Z = \{(A_i \in \C^{n\times n}, B_i \in
  \C^{n\times 1}) : i\in\Ind\}$ is said to be \emph{SLASF (Solvable
    Lie Algebra with Stability by Feedback)} if there exist $K_i \in
  \C^{1\times n}$ such that $A_i^\CL$ as in (\ref{eq:49}) generate a
  solvable Lie algebra and satisfy $\srad(A_i^\CL)<1$. If $\Z$ is
  SLASF, we say that $K_i \in \C^{1\times n}$ are \emph{compatible
    with $\Z$} if $A_i^\CL$ as in (\ref{eq:49}) generate a solvable
  Lie algebra and satisfy $\srad(A_i^\CL)<1$. If $K_i$ exist so that,
  in addition, $\srad(A_i^\CL)\le 1-\gamma$ for some $0< \gamma\le 1$,
  we say that $\Z$ is $\gamma$-SLASF and refer to the corresponding
  $K_i$ as $\gamma$-compatible with $\Z$.
\end{defn}
We next state a version of our previous result of \cite{HBF2009:_CDC}
for the specific case of single-input systems as follows.

\begin{thm}
  \label{thm:prevsic}
  Let $\Z = \{(A_i \in \C^{n\times n}, B_i \in \C^{n\times 1}) :
  i\in\Ind\}$, consider Algorithm~\ref{alg:main} for some suitable
  choice of $\epsilon_c^\ell$ and $\epsilon_d^\ell$, and let $\Z^\ell
  = \{(A_i^\ell,B_i^\ell) : i\in\Ind\}$ and $0<\gamma\le 1$. Then, $\Z$
  is $\gamma$-SLASF if and only if $\Z^\ell$ is $\gamma$-CEAS and
  $v_1^\ell$, $F_i^\ell$ are $\gamma$-compatible with $\Z^\ell$ for
  $\ell=1,\ldots,n$.
\end{thm}

Suppose that $\hat\Z^\ell = \{(\hat A_i^\ell,\hat B_i^\ell) :
i\in\Ind\}$ is CEAS. Whenever $\hat F_i^\ell$, $\hat v_1^\ell$ are
compatible with $\hat\Z^\ell$ and a unitary matrix having $\hat
v_1^\ell$ as its first column is given:
\begin{gather}
  \label{eq:74}
  \big[\hat v_1^\ell | \hat v_2^\ell | \cdots | \hat
  v_{n-\ell+1}^\ell\big],
\end{gather}
then $\hat A_i^{\ell,\CL}$, $\hat U_{\ell+1}$, $\hat A_i^{\ell+1}$ and
$\hat B_i^{\ell+1}$ will denote the matrices given by (\ref{eq:33}),
(\ref{eq:99}), (\ref{eq:100}) and (\ref{eq:6}), respectively, when the
hatted matrices are employed.

The main result of this section, namely Theorem~\ref{thm:KiCQLF},
requires the following preliminary theorem.
\begin{thm}
  \label{thm:num1}
  Let $\hat\Z^\ell = \{(\hat A_i^\ell,\hat B_i^\ell) : i\in\Ind\}$ be
  $\gamma$-CEAS and controllable, with $\hat A_i^\ell\in\C^{n_r \times
    n_r}$, $\hat B_i^\ell \in \C^{n_r \times 1}$, $n_r > 1$, $0<\gamma
  \le 1$. Consider the following sets
  \begin{align}
    \label{eq:setT}
    \hat\Tb^\ell_{i,\gamma} &= \{v \in \Sb_1 :  
    \exists F\in\C^{1\times n_r}, \lambda\in\C \text{ such that } 
    (\hat A_i^\ell + \hat B_i^\ell F)v = \lambda v \text{ with }
    |\lambda| \le 1-\gamma \},\\
    \label{eq:72}
    \hat\Tb^\ell_\gamma &= \bigcap_{i=1}^\s \hat\Tb^\ell_{i,\gamma},
  \end{align}
  and the following quantity
  \begin{align}
    \label{eq:48}
    \epsilon_d^{\ell,\star} &\dfn \inf_{v\in\hat\Tb_{\gamma}^\ell}
    min_{i\in\Ind} \dist(v,\img \hat B_i^\ell).
  \end{align}
  Then, $\epsilon_d^{\ell,\star} > 0$ and each $0 < \epsilon_c^\ell <
  \gamma$ and $0 < \epsilon_d^\ell < \epsilon_d^{\ell,\star}$ ensure
  that for every $\epsilon>0$ there exists a corresponding $\delta>0$
  so that for each $A_i^\ell$, $B_i^\ell$ satisfying
    \begin{equation*}
      \left\lbrace
        \begin{matrix}
          \| \hat A_i^\ell - A_i^\ell \| < \delta,\\
          \| \hat B_i^\ell - B_i^\ell \| < \delta,
        \end{matrix} \right.
    \end{equation*}
    there exist $\hat v_1^\ell \in \hat\Sb(\epsilon_c^\ell/2,
    \epsilon_d^\ell/2)$ and $\hat F_i^\ell$ compatible with
    $\hat\Z^\ell$ ($\hat v_1^\ell$, $\hat F_i^\ell$ may depend on the
    specific $A_i^\ell$, $B_i^\ell$), and a unitary matrix
    (\ref{eq:74}) that cause
    \begin{enumerate}[i)]
    \item \begin{align}
        \label{eq:57}
        &\| \hat v_1^\ell - v_1^\ell \| < \epsilon,\\
        \label{eq:85}
        &\| \hat F_i^\ell - F_i^\ell \| < \epsilon,
      \end{align}
      where $v_1^\ell$ and $F_i^\ell$ are the output of
      Procedure~$\cea$ with $A_i^\ell$, $B_i^\ell$,
      $\epsilon_c=\epsilon_c^\ell$, and $\epsilon_d=\epsilon_d^\ell$
      as inputs.\label{item:5}
    \item \begin{align} \notag
        \| \hat U_{\ell+1} - U_{\ell+1} \| &< \epsilon,\\
        \label{eq:38}
        \| \hat A_i^{\ell+1} - A_i^{\ell+1} \| &< \epsilon,\\
        \label{eq:98}  
        \| \hat B_i^{\ell+1} - B_i^{\ell+1} \| &< \epsilon,
      \end{align}
      where $U_{\ell+1}$, $A_i^{\ell+1}$ and $B_i^{\ell+1}$ are the
      matrices computed at iteration $\ell$ of Algorithm~\ref{alg:main}
      from $A_i^\ell$ and $B_i^\ell$, with $v_1^\ell$ and $F_i^\ell$
      as above.\label{item:6}
    \end{enumerate}
\end{thm}

Theorem~\ref{thm:num1}~\ref{item:5}) establishes that if the matrices
$A_i^\ell$ and $B_i^\ell$ given as inputs to Procedure~$\cea$ are
sufficiently close to some $\hat A_i^\ell$ and $\hat B_i^\ell$ which
form a CEAS set $\hat\Z^\ell$, then the vector $v_1^\ell$ and feedback
matrices $F_i^\ell$ computed by such procedure will be as close as
desired to some $\hat v_1^\ell$ and $\hat F_i^\ell$ compatible with
$\hat\Z^\ell$.  In general, whether the (given) matrices $A_i^\ell$,
$B_i^\ell$ are sufficiently close to some $\hat A_i^\ell$, $\hat
B_i^\ell$ with the required property will not be known. However, the
significance of this result lies precisely in the fact that it
establishes a type of continuity relation between the result of the
procedure and an ``exact'' result, even if the latter result is not
known. In addition, such continuity justifies the numerical
implementation of the procedure, since numerical computation will
always yield an approximate result. 

In broad terms, Theorem~\ref{thm:num1}~\ref{item:6}) shows that if, at
step $\ell$ of Algorithm~\ref{alg:main}, the matrices $A_i^\ell$ and
$B_i^\ell$ are sufficiently close to some ``exact'' ones, then the
same will happen at step $\ell+1$. The proof of
Theorem~\ref{thm:num1} is highly non-trivial and given in
Appendix~\ref{sec:proof-theorem}.

We are now ready to state the main result of the paper.
\begin{thm}
  \label{thm:KiCQLF}
  Let $\hat\Z = \{(\hat A_i\in\C^{n \times n},\hat B_i\in \C^{n \times
    1}) : i\in\Ind\}$ be SLASF and controllable. Then, there exist
  $\epsilon_c^\star,\epsilon_d^\star>0$ such that each $0 < \epsilon_c
  < \epsilon_c^\star$ and $0 < \epsilon_d < \epsilon_d^\star$ ensure
  that
  \begin{enumerate}[i)]
  \item For every $\epsilon>0$ there exists a corresponding $\delta>0$
    so that for each $A_i,B_i$ satisfying
    \begin{align}
      \label{eq:4}
      \| \hat A_i - A_i \| &< \delta,\\
      \label{eq:7}
      \| \hat B_i - B_i \| &< \delta,
    \end{align}
    there exist $\hat K_i$ compatible with $\hat\Z$ ($\hat K_i$ may
    depend on the specific $A_i$, $B_i$) that cause $\srad(\hat
    A_i^\CL)\le 1 - \epsilon_c/2$ and
    \begin{equation}
      \label{eq:2}
      \| \hat A_i^\CL - A_i^\CL \| < \epsilon,
    \end{equation}
    where $\hat A_i^\CL=\hat A_i + \hat B_i \hat K_i$, and $A_i^\CL$
    satisfies (\ref{eq:49}) with $K_i$ obtained as output of
    Algorithm~\ref{alg:main} with $A_i, B_i, \epsilon_c$, and
    $\epsilon_d$ as inputs.\label{item:4}
  \item There exists $\epsilon>0$ for which the closed-loop DTSS
    (\ref{eq:26})--(\ref{eq:49}) admits a CQLF, provided
    (\ref{eq:4})--(\ref{eq:7}) are satisfied with $\delta$
    corresponding to $\epsilon$ as in \ref{item:4})
    above.\label{item:7}
  \end{enumerate}
\end{thm}

Theorem~\ref{thm:KiCQLF} establishes that Algorithm~\ref{alg:main}
will compute suitable feedback matrices not only in the ``exact'' case
when the given $A_i$, $B_i$ form a SLASF set, but also when they are
close to other (possibly unknown) matrices $\hat A_i$, $\hat B_i$ with
such property. This justifies the use of Algorithm~\ref{alg:main} for
control design, since it gives a kind of robustness result for the
feedback matrices computed by the algorithm. Theorem~\ref{thm:KiCQLF}
also establishes that suitable feedback matrices will be computed for
all positive $\epsilon_c$ and $\epsilon_d$ respectively less than
$\epsilon_c^\star$ and $\epsilon_d^\star$. The latter quantities may
be not known, since they depend on the possibly unknown $\hat A_i$ and
$\hat B_i$. Consequently, in theory the parameters $\epsilon_c$ and
$\epsilon_d$ should be selected as small as computationally
possible. In practice, however, there is a tradeoff in the selection
of $\epsilon_c$ since the smaller $\epsilon_c$, the higher the chances
of Algorithm~\ref{alg:main} yielding unsuitable feedback matrices when
more than one common eigenvector could be assigned by feedback with
one of them having unstable corresponding eigenvalues and another
having stable ones (Procedure~$\cea$ could select a local optimiser at
the boundary of the constraint set instead of a global optimiser in
its interior).

  After executing Algorithm~\ref{alg:main} to compute feedback
  matrices $K_i$, we can check whether the closed-loop DTSS
  (\ref{eq:26})--(\ref{eq:49}) admits a CQLF by solving the following
  LMIs
  \begin{equation}
    \label{eq:10}
    P=P^T > 0,\quad P - (A_i^\CL)^T P A_i^\CL > 0, \text{ for }i\in\Ind.
  \end{equation}
  Note that, in this case, LMIs are used only to check whether a CQLF
  for the closed-loop system exists, and not for feedback design. If
  these LMIs are feasible, then not only the closed-loop DTSS admits a
  CQLF but also we have structural information on the DTSS since $K_i$
  are such that the $A_i^\CL$ are suitably close to being
  simultaneously triangularisable.



\section{Examples}
\label{sec:examples}

We next provide some numerical examples to illustrate the advantages
and limitations of the proposed feedback design strategy. For the
numerical implementation of Procedure~$\cea$, a feasible vector is
first sought using \textsc{Matlab}$^{\tiny\textregistered}$
\textsc{Optimization Toolbox} function \texttt{fgoalattain}. If such
vector is found, it is passed as initial point to the optimisation,
implemented via the function \texttt{fmincon}.

\subsection{Randomly created DTSS}
\label{sec:rand-creat-dtss}


The following subsystems were created randomly but such that
$\srad(A_1)<1$ and $\srad(A_2^{-1})<1$.
\begin{alignat}{2}
  \label{eq:11}
  A_1 &= \left[
    \begin{smallmatrix}
      0.574  & 0.074  & 0.089\\
      0.074  & 0.572  &-0.091\\
      0.089  &-0.091  & 0.538
    \end{smallmatrix}\right],&\quad
  B_1 &= \left[
    \begin{smallmatrix}
      -0.038\\ 0.327\\ 0.175
    \end{smallmatrix}\right],\\
  \label{eq:20}
  A_2 &= \left[
    \begin{smallmatrix}
      -0.737 & 0.386  &-1.680\\
       1.351 & 0.638  & 0.035\\
       1.071 &-1.295  &-0.936
    \end{smallmatrix}\right],&\quad
  B_2 &= \left[
    \begin{smallmatrix}
      0\\ 0.114\\ 1.067
    \end{smallmatrix}\right].
\end{alignat}
Executing Algorithm~\ref{alg:main} choosing $\epsilon_c = \epsilon_d =
10^{-4}$ yields a feasible optimisation at every iteration and returns
\begin{gather}
  \label{eq:12}
  \begin{smallmatrix}
    K_1 &= \left[
      \begin{smallmatrix}
        -3.6480 &-7.2304 & 8.7751
      \end{smallmatrix}\right],\\
    K_2 &= \left[
      \begin{smallmatrix}
        -0.3159 & 2.0235 & 0.2695
      \end{smallmatrix}\right],
  \end{smallmatrix}\;
  \small U=\left[
    \begin{smallmatrix}
       0.4647 &0.8287 &-0.3120\\
      -0.7770 &0.2126 &-0.5925\\
      -0.4246 &0.5178 & 0.7427
    \end{smallmatrix}\right].
\end{gather}
It can be shown that LMIs (\ref{eq:10}) are feasible and hence the
closed-loop DTSS (\ref{eq:26})--(\ref{eq:49}) with $K_i$ as in
(\ref{eq:12}) admits a CQLF and is hence stable under arbitrary
switching. In this case, it can be checked that $U^* A_i^\CL U$ are
not upper triangular but are close in the sense that the entries below
the main diagonal are small. Therefore, the use of feedback matrices
(\ref{eq:12}) designed via Algorithm~\ref{alg:main} provides some
insight into the structure of the closed-loop DTSS.

On the other hand, solution of the LMIs (\ref{eq:13}) for both
feedback design and CQLF computation, which yields
\begin{equation*}
    K_1 = \left[
      \begin{smallmatrix}
        -1.2267 &-0.7211 &-1.8731
      \end{smallmatrix}\right],\quad
    K_2 = \left[
      \begin{smallmatrix}
        -0.5140 &1.3826  &1.1613
      \end{smallmatrix}\right],
\end{equation*}
is guaranteed to produce a closed-loop DTSS stable under arbitrary
switching but provides no structural information.

\subsection{DTSS with no CQLF}
\label{sec:non-stab-dtss}

Consider the systems
\begin{equation*}
  A_1 = \left[
    \begin{smallmatrix}
      0.5 &\alpha\\
      0   &0.5
    \end{smallmatrix}\right],
  B_1 = \left[
    \begin{smallmatrix}
      0\\ 1
    \end{smallmatrix}\right],\qquad
  A_2 = \left[
    \begin{smallmatrix}
      0.5    &0\\
      \alpha &0.5
    \end{smallmatrix}\right],
  B_2 = \left[
    \begin{smallmatrix}
      1\\ 0
    \end{smallmatrix}\right].
\end{equation*}
For $\alpha=1.5$, the LMIs (\ref{eq:13}) are not feasible. Therefore,
no CQLF exists for this DTSS. Executing Algorithm~\ref{alg:main} with
$\epsilon_c = \epsilon_d = 10^{-4}$ yields an infeasible optimisation
at the first iteration, $\Sb(\epsilon_c,\epsilon_d)=\emptyset$.

For $\alpha=1.4999$, the LMIs (\ref{eq:13}) are feasible. In
this case, Algorithm~\ref{alg:main} for the selected $\epsilon_c$,
$\epsilon_d$ yields an infeasible optimisation. However, reducing
$\epsilon_c$ to $10^{-5}$ allows the algorithm again to compute
suitable feedback matrices.

\subsection{CQLF exists but Algorithm~\ref{alg:main} fails}
\label{sec:alg}


Consider again the DTSS (\ref{eq:11})--(\ref{eq:20}), with the
addition of the subsystem
\begin{equation*}
  A_3 = \left[
    \begin{smallmatrix}
       0.352  &0.159 &-1.129\\
       0.159  &0     & 0.262\\
      -1.129  &0.262 &-0.705
    \end{smallmatrix}\right],\quad
  B_3 = \left[
    \begin{smallmatrix}
      -0.433\\ 0\\ 0
    \end{smallmatrix}\right].
\end{equation*}
Algorithm~\ref{alg:main} for $\epsilon_c = \epsilon_d = 10^{-4}$ yields 
\begin{gather}
  \label{eq:22}
  \begin{tabular}[c]{l}
    $K_1 = \left[
      \begin{smallmatrix}
        -15.3542 &3.8969 &-11.3814
      \end{smallmatrix}\right],$\\
    $K_2 = \left[
      \begin{smallmatrix}
        0.0734   &0.9747 &2.7288
      \end{smallmatrix}\right],$\\
    $K_3 = \left[
      \begin{smallmatrix}
        -1.3542  &0.8334 &-4.5001
      \end{smallmatrix}\right],$
  \end{tabular}\\
  \notag
  U = \left[
    \begin{smallmatrix}
      0.1662 + 0.0234j &-0.0918 - 0.6508j &-0.7347\\
      0.9744 + 0.1374j & 0.0093 + 0.0662j & 0.1650\\
      0.0587 + 0.0083j & 0.1049 + 0.7433j &-0.6580
    \end{smallmatrix}\right].
\end{gather}
However, the optimisation informs that there is an active inequality,
corresponding to the stability constraint (\ref{eq:61}). In this case,
the LMIs (\ref{eq:10}) are not feasible and hence no CQLF exists when
the feedback matrices (\ref{eq:22}) are employed.

On the other hand, the LMIs (\ref{eq:13}) are feasible and hence other
feedback matrices may indeed produce a closed-loop DTSS with a CQLF.

\section{Conclusions}
\label{sec:conclusions}

This paper complements the theoretical results in \cite{HBF2009:_CDC},
and contributes to furthering the understanding of fundamental system
structure in feedback stabilisation of DTSS. We have presented a
numerical implementation of a feedback design strategy for DTSS based
on Lie-algebraic solvability. The proposed strategy seeks feedback
matrices to achieve a closed-loop system structure that
\emph{approximates} that required to satisfy such Lie-algebraic
stability criteria.

The main theoretical contribution of the paper establishes that if a
system for which the Lie-algebraic conditions considered exists in a
suitably small neighbourhood of the given system data, then our
implementation will find feedback matrices so that the corresponding
closed-loop DTSS admits a CQLF \emph{even if the considered
  Lie-algebraic conditions are not met by the given system}. However,
since the existence of such feasible ``exact'' system is in general
unknown, the resulting closed-loop system is not guaranteed to admit a
CQLF but the latter may be checked with a set of \emph{informed} LMIs
built with the computed closed-loop matrices. Whether such ``exact''
system exists suitably close to a given system has not been discussed,
and remains a topic for further research. Future work will also
consider extensions to multiple input systems, which are nontrivial
and will possibly require the consideration of controllability indices
in the algorithm.


\appendix

\section{Appendix}
\label{sec:my-appendix}

\subsection{Proof of Theorem~\ref{thm:num1}}
\label{sec:proof-theorem}

Throughout this proof, a hatted expression denotes the expression
given by the corresponding equations when matrices $\hat A_i^\ell$,
$\hat B_i^\ell$ are substituted for $A_i^\ell$, $B_i^\ell$. For
example, $\hat\H_i(v) = (vv^* - \id)\hat B_i^\ell$.

We begin by establishing that $\epsilon_d^{\ell,\star} > 0$.
The set $\hat\Tb^\ell_{\gamma}$ is the set of all unit vectors
that are $\gamma$-compatible with $\hat\Z^\ell$, and since
$\hat\Z^\ell$ is $\gamma$-CEAS, then $\hat\Tb^\ell_{\gamma} \neq
\emptyset$. Since $\hat\Z^\ell$ is controllable, then $(\hat
A_i^\ell,\hat B_i^\ell)$ is controllable and hence $\hat B_i^\ell \neq
0$.
\begin{claim}
  \label{clm:Tclosed}
  The set $\hat\Tb^\ell_{\gamma}$ is compact.
\end{claim}
\begin{proof}
  The set $\hat\Tb^\ell_{i,\gamma}$ is bounded since
  $\hat\Tb^\ell_{i,\gamma}\subset\Sb_1$ by definition. The set
  $\hat\Tb^\ell_{i,\gamma}$ can be equivalently defined as
  \begin{equation}
    \label{eq:Talt}
    \hat\Tb^\ell_{i,\gamma} = \{v \in \Sb_1 : 
    \exists \lambda\in\C \text{ such  that } 
    \hat P_i^\ell(\lambda \id - \hat A_i^\ell) v = 0, \text{ with }|\lambda| \le 1-\gamma \},
  \end{equation}
  where we have defined
  \begin{equation}
    \label{eq:P}
    \hat P_i^\ell = \left[\id - \hat B_i^\ell\left((\hat B_i^\ell)^* \hat B_i^\ell\right)^{-1}(\hat B_i^\ell)^*\right].
  \end{equation}
  Note that since $0\neq \hat B_i^\ell \in \C^{n_r \times 1}$, then
  $\hat P_i^\ell$ is well-defined. Consider a sequence
  $\{v_k\}_{k=0}^\infty$ such that $v_k \in \hat\Tb_{i,\gamma}^\ell$
  for all $k\ge 0$ and $\lim_{k\rightarrow\infty} v_k = v$. Since
  $\|v_k\| = 1$ for all $k\ge 0$ and by the continuity of norms, then
  $\|v\|=1$ necessarily. For every $k\ge 0$, we have
  \begin{equation}
    \label{eq:24}
    \hat P_i^\ell(\lambda_k \id - \hat A_i^\ell) v_k = 0,
  \end{equation}
  for some $\lambda_k \in \C$ with $|\lambda_k|\le 1-\gamma$. From (\ref{eq:24}),
  \begin{equation}
    \lim_{k\rightarrow\infty} \hat P_i^\ell v_k\lambda_k =
    \lim_{k\rightarrow\infty} \left[\hat P_i^\ell (v_k-v)\lambda_k +
    \hat P_i^\ell v\lambda_k\right] = \hat P_i^\ell \hat A_i^\ell v.
  \end{equation}
  Since $\lambda_k$ is bounded and $\lim_{k\rightarrow\infty} v_k=v$,
  it follows that
  \begin{equation}
    \lim_{k\rightarrow\infty} \hat P_i^\ell v\lambda_k = \hat P_i^\ell
    \hat A_i^\ell v.
  \end{equation}
  If $\hat P_i^\ell v\neq 0$, then $\lim_{k\rightarrow\infty} \lambda_k = \lambda$
  with $|\lambda|\le 1-\gamma$. If $\hat P_i^\ell v = 0$, then $\hat
  P_i^\ell \hat A_i^\ell v=0$. In either case,
  $v\in\hat\Tb_{i,\gamma}^\ell$ and hence $\hat\Tb_{i,\gamma}^\ell$ is
  closed. Therefore, $\hat\Tb_{\gamma}^\ell$ is closed since it is the
  intersection of a finite number of closed sets.
\end{proof}

The quantity $\epsilon_d^{\ell,\star}$ as defined in (\ref{eq:48}) is
the infimum, over all vectors $\gamma$-compatible with $\hat\Z^\ell$,
of the minimum of the distance between such vectors and $\img \hat
B_i^\ell$. Since $(\hat A_i^\ell,\hat B_i^\ell)$ is controllable and
$n_r>1$, then Lemma~\ref{lem:SIgood} implies that $v\notin\img\hat
B_i^\ell$ for $i\in\Ind$ and every $v \in
\hat\Tb_{\gamma}^\ell$. Therefore, for every $v \in
\hat\Tb_{\gamma}^\ell$, we have $\min_{i\in\Ind} \dist(v,\img \hat
B_i^\ell) > 0$ because $\img \hat B_i^\ell$ is closed for
$i\in\Ind$. Since $\hat\Tb_{\gamma}^\ell$ is compact and
$\min_{i\in\Ind} \dist(v,\img \hat B_i^\ell)$ is continuous and
positive at every $v\in\hat\Tb_{\gamma}^\ell$, it follows that
$\min_{i\in\Ind} \dist(v,\img \hat B_i^\ell)$ achieves a minimum on
$\hat\Tb_{\gamma}^\ell$ and hence $\epsilon_d^{\ell,\star} > 0$.

\subsubsection{Proof of Theorem~\ref{thm:num1}~\ref{item:5})}

\begin{claim}
  \label{clm:eps}
  For every $0<\epsilon_c^\ell<\gamma$ and
  $0<\epsilon_d^\ell<\epsilon_d^{\ell,\star}$, there exists
  $\delta_2>0$ such that $\hat\Tb_{\gamma}^\ell \cap
  \Sb(\epsilon_c^{\ell}, \epsilon_d^{\ell}) \neq \emptyset$ for all
  $A_i^\ell$, $B_i^\ell$ satisfying
  \begin{equation}
    \label{eq:70}
        \| \hat A_i^\ell - A_i^\ell \| < \delta_2,\quad
        \| \hat B_i^\ell - B_i^\ell \| < \delta_2.
  \end{equation}
\end{claim}
\begin{proof}
Let $\vv$ and $\hat G_i$ be $\gamma$-compatible with $\hat\Z^\ell$ and
$\|\vv\|=1$. Note that $\vv\in\hat\Tb_\gamma^\ell$.
  By (\ref{eq:48}), $\hat G_i$ and $\vv$ satisfy, for $i\in\Ind$,
  \begin{align}
    \label{eq:40}
    &(\hat A_i^\ell + \hat B_i^\ell \hat G_i)\vv = \lambda_i \vv,\\
    \label{eq:59}
    &|\lambda_i| \le 1-\gamma < 1,\\
    \label{eq:58}
    &\dist(\vv,\img \hat B_i^\ell) \ge \epsilon_d^{\ell,\star} > 0.
  \end{align}
  Consider $\hat{\H}_i(\vv)$ and $\H_i(\vv)$ from (\ref{eq:18}). By
  (\ref{eq:58}) and since $\hat B_i^\ell$ have full rank, then
  $\hat{\H}_i(\vv)$ has full rank. Then, for all $\delta_0 > 0$
  sufficiently small, ${\H}_i(\vv)$ has full rank whenever $\|\hat
  B_i^\ell - B_i^\ell \| < \delta_0$. Whenever the latter holds, the
  expression (\ref{eq:30}) is continuous on the entries of $A_i^\ell$
  and $B_i^\ell$. Then, given $\epsilon_0 > 0$, we can find $0 <
  \delta_1 \le \delta_0$ such that
  \begin{gather}
    \notag
    \| \hat A_i^{\ell,\CL}(\vv) \vv - A_i^{\ell,\CL}(\vv) \vv \| < \epsilon_0\\
    \label{eq:68}
    \text{whenever}\quad \left\lbrace
      \begin{matrix}
        \| \hat A_i^\ell - A_i^\ell \| < \delta_{1},\\
        \| \hat B_i^\ell - B_i^\ell \| < \delta_{1}.
      \end{matrix} \right.
  \end{gather}
  By (\ref{eq:40}) and Lemma~\ref{lem:comeig}-\ref{item:3}), we have
  $\hat A_i^{\ell,\CL}(\vv) \vv = \lambda_i \vv$. Since $\|\vv\|=1$,
  then
  \begin{equation*}
    \| \hat A_i^{\ell,\CL}(\vv) \vv \| = |\lambda_i| \le
    1-\gamma < 1
  \end{equation*}
  by (\ref{eq:59}), and we can select $\epsilon_0>0$ small enough so
  that $\| A_i^{\ell,\CL}(\vv) \vv \| \le 1 - \epsilon_c^{\ell}$ and
  hence $\vv\in\Sb_2(\epsilon_c^{\ell})$ whenever (\ref{eq:68})
  holds. 
  For $0 \le a < \delta_1$, define
  \begin{equation*}
    \md(a) \dfn \min_{i\in\Ind}\: \inf_{B_i^\ell:\|\hat B_i^\ell - B_i^\ell \|
      \le a} \dist(\vv,\img B_i^\ell)
  \end{equation*}
  and note that, by the continuity of $\dist(\vv,\img B_i^\ell)$ on
  the entries of $B_i^\ell$ whenever $B_i^\ell$ has full rank, and
  since by (\ref{eq:58}) $\md(0) \ge \epsilon_d^{\ell,\star} > 0$,
  there exists $\delta_2 > 0$ sufficiently small for which
  $\md(\delta_2) > \epsilon_d^\ell > 0$. Therefore, for such
  $\delta_2$ we have $\vv\in\Sb_3(\epsilon_d^{\ell})$ and hence
  $\vv\in\Sb(\epsilon_c^\ell,\epsilon_d^\ell)$ whenever (\ref{eq:70})
  holds.
\end{proof}

\begin{claim}
  \label{clm:uniform}
  Consider $0 < \epsilon_c^\ell < \gamma$ and $0 <
  \epsilon_d^\ell < \epsilon_d^{\ell,\star}$. Then, there exists
  $\delta_3>0$ so that
  \begin{enumerate}[i)]
  \item (\ref{eq:47})--(\ref{eq:43}) hold for every $A_i^\ell$,
    $B_i^\ell$ satisfying (\ref{eq:46}).
    \begin{gather}
      \label{eq:47}
      \sup_{v\in\Sb(\epsilon_c^\ell,\epsilon_d^\ell)} \max_{i\in\Ind} \| \hat
      A_i^{\ell,\CL}(v) v \| \le 1 - \epsilon_c^\ell/2,\\
      \label{eq:43}
      \inf_{v\in\Sb(\epsilon_c^\ell,\epsilon_d^\ell)} \min_{i\in\Ind} \dist(v,\img
      \hat B_i^\ell) \ge \epsilon_d^\ell/2,\\
      \label{eq:46}
      \| \hat A_i^\ell - A_i^\ell \| < \delta_3,\quad
      \| \hat B_i^\ell - B_i^\ell \| < \delta_3.
    \end{gather}\label{item:12}
  \item For every $\epsilon_2>0$, there exists $0<\delta_4<\delta_3$
    such that (\ref{eq:51}) holds for every $A_i^\ell$,
    $B_i^\ell$ satisfying (\ref{eq:71}).
    \begin{gather}
      \label{eq:51}
      | \hat J(v) - J(v) | < \epsilon_2,
      \quad\text{for all } v\in\Sb(\epsilon_c^\ell,\epsilon_d^\ell)\\
      \label{eq:71}
      \| \hat A_i^\ell - A_i^\ell \| < \delta_4,\quad
      \| \hat B_i^\ell - B_i^\ell \| < \delta_4.
    \end{gather}\label{item:13}
  \end{enumerate}
\end{claim}
\begin{proof}
  \ref{item:12}) Consider the functions
  \begin{gather*}
    f\big(\{A_i^\ell, B_i^\ell : i\in\Ind\}\big) \dfn
    \sup_{v\in\Sb(\epsilon_c^\ell,\epsilon_d^\ell)} \max_{i\in\Ind} \| \hat
    A_i^{\ell,\CL}(v) v - A_i^{\ell,\CL}(v) v \|, \\
    g\big(\{A_i^\ell, B_i^\ell : i\in\Ind\}\big) \dfn 
    \sup_{v\in\Sb(\epsilon_c^\ell,\epsilon_d^\ell)} \max_{i\in\Ind}
    |\dist(v,\img \hat B_i^\ell) - \dist(v,\img B_i^\ell)|,
  \end{gather*}
  which are non-negative whenever well-defined. Note that
  \begin{equation*}
    f\big(\{\hat A_i^\ell, \hat B_i^\ell : i\in\Ind\}\big) = 0 =
    g\big(\{\hat A_i^\ell, \hat B_i^\ell : i\in\Ind\}\big),
  \end{equation*}
  and that, since $\hat B_i^\ell$ full column rank, $f,g$ are
  continuous on the entries of $A_i^\ell$, $B_i^\ell$ and well-defined
  whenever (\ref{eq:46}) holds for some $\delta_3$ sufficiently
  small. Therefore, for every $\epsilon_1 > 0$, we can find a
  corresponding $\delta_3 > 0$ so that
  \begin{equation*}
    f\big(\{A_i^\ell, B_i^\ell : i\in\Ind\}\big) \le
    \epsilon_1,\quad
    g\big(\{A_i^\ell, B_i^\ell : i\in\Ind\}\big) \le \epsilon_1
  \end{equation*}
  whenever (\ref{eq:46}) holds. Select $\epsilon_1 \le
  \min\{\epsilon_c^\ell/2, \epsilon_d^\ell/2\}$, and take the corresponding
  $\delta_3$. Thus, (\ref{eq:47})--(\ref{eq:43}) hold whenever
  (\ref{eq:46}) is true.

  \ref{item:13}) Consider the function
  \begin{equation*}
    h\big(\{A_i^\ell, B_i^\ell : i\in\Ind\}\big) \dfn
    \sup_{v\in\Sb(\epsilon_c^\ell,\epsilon_d^\ell)} | \hat J(v) - J(v) |.
  \end{equation*}
  Note that, because of (\ref{eq:43}), $h$ is well-defined,
  non-negative and continuous whenever (\ref{eq:46}) holds and that
  \begin{equation*}
    h\big(\{\hat A_i^\ell, \hat B_i^\ell : i\in\Ind\}\big) = 0.
  \end{equation*}
  Therefore, given $\epsilon_2 > 0$ we can find $0 < \delta_4 <
  \delta_3$, such that (\ref{eq:51}) holds whenever (\ref{eq:71}) is
  true.
\end{proof}


Fix $\epsilon_c^\ell$ and $\epsilon_d^\ell$ so that $0 <
\epsilon_c^\ell < \gamma$ and $0 < \epsilon_d^\ell <
\epsilon_d^{\ell,\star}$. Consider $A_i^\ell$ and $B_i^\ell$
satisfying both (\ref{eq:70}) with $\delta_2$ as given by
Claim~\ref{clm:eps} and (\ref{eq:46}) with $\delta_3$ from
Claim~\ref{clm:uniform}-\ref{item:12}). Let
$v_1^\ell\in\Sb(\epsilon_c^\ell,\epsilon_d^\ell)$ be a (global)
minimiser of $J(v)$ subject to
$v\in\Sb(\epsilon_c^\ell,\epsilon_d^\ell)$ [such minimiser exists
because $J$ is continuous on $\Sb(\epsilon_c^\ell,\epsilon_d^\ell)$
and $\Sb(\epsilon_c^\ell,\epsilon_d^\ell)$ is compact and
nonempty]. Define $\hat\epsilon_c \dfn \epsilon_c^\ell/2$ and
$\hat\epsilon_d \dfn \epsilon_d^\ell/2$, pick $\vv \in
\Sb(\epsilon_c^\ell,\epsilon_d^\ell) \cap \hat\Tb_\gamma^\ell$, and
note that
  \begin{align*}
    0 < \hat\epsilon_c &\le 1 - \max_{i\in\Ind}
    \max\Big\{\|\hat A_i^{\ell,\CL}(v_1^\ell) v_1^\ell \|,
    \|\hat A_i^{\ell,\CL}(\vv) \vv \|\Big\},\text{ and}\\
    0 < \hat\epsilon_d &\le \min_{i\in\Ind}
    \min\Big\{\dist(v_1^\ell,\img \hat B_i^\ell), \dist(\vv,\img \hat
    B_i^\ell)\Big\},
  \end{align*}
  by (\ref{eq:47})--(\ref{eq:43}) from
  Claim~\ref{clm:uniform}-\ref{item:12}), and since $v_1^\ell \in
  \Sb(\epsilon_c^\ell,\epsilon_d^\ell)$ and $\vv \in
  \Sb(\epsilon_c^\ell,\epsilon_d^\ell)$. Thus, we have $v_1^\ell \in
  \hat\Sb(\hat\epsilon_c,\hat\epsilon_d)$, $\vv \in
  \hat\Sb(\hat\epsilon_c,\hat\epsilon_d)$ and by
  Lemma~\ref{lem:comeig}-\ref{item:2}), we know $\hat J(\vv)=0$. Since
  $v_1^\ell$ is a minimiser within
  $\Sb(\epsilon_c^\ell,\epsilon_d^\ell)$, and
  $\vv\in\Sb(\epsilon_c^\ell,\epsilon_d^\ell)$, then $J(v_1^\ell) \le
  J(\vv)$. Thus, using Claim~\ref{clm:uniform}-\ref{item:13}), we have
  \begin{gather}
    \label{eq:50}
    J(v_1^\ell) \le J(\vv) = | \hat J(\vv) - J(\vv) | < \epsilon_2,\\
    \label{eq:56}
    \text{whenever}\quad
        \| \hat A_i^\ell - A_i^\ell \| < \delta_4,\quad
        \| \hat B_i^\ell - B_i^\ell \| < \delta_4.
  \end{gather}
  Also, using (\ref{eq:50}) and Claim~\ref{clm:uniform}-\ref{item:13}), under
  (\ref{eq:56}) we will have
  \begin{equation*}
    \hat J(v_1^\ell) \le |\hat J(v_1^\ell) - J(v_1^\ell) | 
    + J(v_1^\ell) < 2\epsilon_2.
  \end{equation*}
  Define $ \hat\V_0 \dfn \{v \in
  \hat\Sb(\hat\epsilon_c,\hat\epsilon_d) \cap
  \Sb(\epsilon_c^\ell,\epsilon_d^\ell) : \hat J(v) = 0\}$
  and note that $\vv\in\hat\V_0$. For each $\epsilon_3 > 0$, consider
  \begin{equation*}
    \B(\epsilon_3) \dfn \{ v \in
    \hat\Sb(\hat\epsilon_c,\hat\epsilon_d) \cap
    \Sb(\epsilon_c^\ell,\epsilon_d^\ell) : \dist(v,\hat\V_0) < \epsilon_3 \}.
  \end{equation*}
  By continuity of $J$ and since both
  $\hat\Sb(\hat\epsilon_c,\hat\epsilon_d)$ and
  $\Sb(\epsilon_c^\ell,\epsilon_d^\ell)$ are compact, for every
  $\epsilon_3 > 0$ we can find $\epsilon_2 > 0$ such that
  \begin{equation}
    \label{eq:55}
    \hat J(v_1^\ell) < 2\epsilon_2 \quad\Rightarrow
    \quad v_1^\ell \in \B(\epsilon_3)
  \end{equation}
  Note that the closedness of $\hat\Sb(\hat\epsilon_c,\hat\epsilon_d)$
  and $\Sb(\epsilon_c^\ell,\epsilon_d^\ell)$, and hence of
  $\hat\Sb(\hat\epsilon_c,\hat\epsilon_d) \cap
  \Sb(\epsilon_c^\ell,\epsilon_d^\ell)$, is key in allowing the
  implication (\ref{eq:55}). $v_1^\ell \in \B(\epsilon_3)$ implies the
  existence of $\hat v_1^\ell \in \hat\V_0$ such that $\| \hat
  v_1^\ell - v_1^\ell \| < \epsilon_3$, and $\hat v_1^\ell \in
  \hat\V_0$ implies that $\hat J(\hat v_1^\ell) = 0$. By
  Lemma~\ref{lem:comeig} and since $\hat v_1^\ell \in
  \hat\Sb(\hat\epsilon_c,\hat\epsilon_d)$, we have $\hat
  A_i^{\ell,\CL}(\hat v_1^\ell) \hat v_1^\ell = \alpha_i \hat
  v_1^\ell$ with $|\alpha_i| \le 1 - \hat\epsilon_c < 1$, which
  establishes that $\hat v_1^\ell$ and $\hat F_i^\ell \dfn \hat
  M_i(\hat v_1^\ell)$ are compatible with $\hat\Z^\ell$.

  Next, since ${\H}_i(v)$ has full rank whenever $v \in
  \hat\Sb(\hat\epsilon_c,\hat\epsilon_d) \cap
  \Sb(\epsilon_c^\ell,\epsilon_d^\ell)$ and (\ref{eq:56}), then $M_i(v)$ is
  continuous on $v$ and on the entries of $A_i^\ell$,
  $B_i^\ell$. Therefore, given $\epsilon > 0$ we can find $\delta_5
  > 0$ such that $\| \hat M_i(\hat v_1^\ell) - M_i(v_1^\ell) \| <
  \epsilon$ whenever
  \begin{equation*}
    \label{eq:86}
    \| \hat A_i^\ell - A_i^\ell \| < \delta_5,\quad
    \| \hat B_i^\ell - B_i^\ell \| < \delta_5,\quad
    \| \hat v_1^\ell - v_1^\ell \| < \delta_5.
  \end{equation*}
  Take $\epsilon_3 = \min\{\epsilon,\delta_5\}$ to select $\delta_4$
  as above. Finally, taking $\delta =
  \min\{\delta_2,\delta_4,\delta_5\}$ concludes the proof of
  Theorem~\ref{thm:num1}~\ref{item:5}).

\subsubsection{Proof of Theorem~\ref{thm:num1}~\ref{item:6})}
\label{sec:proof-theorem-2}

  Since (\ref{eq:100}) is continuous on the entries of $U_{\ell+1}$
  and $A_i^{\ell,\CL}$, given $\epsilon > 0$ we can find $0 < \delta_1
  < \epsilon$ such that (\ref{eq:38}) holds whenever
  \begin{align}
    \label{eq:101}
    \| \hat A_i^{\ell,\CL} - A_i^{\ell,\CL} \| &< \delta_1,\\
    \notag
    \| \hat U_{\ell+1} - U_{\ell+1} \| &< \delta_1.
  \end{align}
  Since (\ref{eq:6}) is continuous on the entries of $B_i^\ell$ and
  $U_{\ell+1}$, then given $\epsilon > 0$ we can find $0 < \delta_2 <
  \delta_1$ so that (\ref{eq:98}) holds whenever
  \begin{align}
    \notag
    \| \hat B_i^\ell - B_i^\ell \| &< \delta_2,\\
    \label{eq:96}
    \| \hat U_{\ell+1} - U_{\ell+1} \| &< \delta_2.
  \end{align}

  Consider square unitary matrices $\hat W = [\hat v_1^\ell | \hat
  U_{\ell+1}]$ and $W = [v_1^\ell | U_{\ell+1}]$. Note that, given
  $\hat v_1^\ell$, $v_1^\ell$ and $U_{\ell+1}$ so that $W$ is square
  and unitary, and $\|\hat v_1^\ell\| = 1$, for every $\delta_2 > 0$
  we can find $\hat U_{\ell+1}$ and $0 < \delta_3 < \delta_2$ so that
  \begin{align}
    \label{eq:97}
    \| \hat W - W \| &< \delta_2
  \end{align}
  whenever $\| \hat v_1^\ell - v_1^\ell \| < \delta_3.$
  Note that (\ref{eq:97}) implies (\ref{eq:96}) since $\hat
  U_{\ell+1}$ and $U_{\ell+1}$ are the last $n-\ell$ columns of $\hat
  W$ and $W$, respectively.
%
%
  Since (\ref{eq:33}) is continuous, given $\delta_1 > 0$ we can find
  $0 < \delta_4 < \delta_3$ so that (\ref{eq:101}) holds whenever
  \begin{gather}
    \notag
    \| \hat A_i^\ell - A_i^\ell \| < \delta_4,\quad
    \| \hat B_i^\ell - B_i^\ell \| < \delta_4,\quad\text{and}\; \\
   \label{eq:105}
    \| \hat F_i^\ell - F_i^\ell \| < \delta_4.
  \end{gather}

  Applying Theorem~\ref{thm:num1}~\ref{item:5}) yields that each $0 <
  \epsilon_c^\ell < \gamma$, $0 < \epsilon_d^\ell <
  \epsilon_d^{\ell,\star}$ ensure that for the given $\delta_4 > 0$,
  we can find $0 < \delta < \delta_4$ so that for each
  $A_i^\ell,B_i^\ell$ satisfying $\| \hat A_i^\ell - A_i^\ell \| <
  \delta$ and $\| \hat B_i^\ell - B_i^\ell \| < \delta$, there exist
  $\hat v_1^\ell \in \hat\Sb(\epsilon_c^\ell/2,\epsilon_d^\ell/2)$ and
  $\hat F_i^\ell$ compatible with $\hat\Z^\ell$ so that (\ref{eq:105})
  holds and $\| \hat v_1^\ell - v_1^\ell \| < \delta_4 < \epsilon$.

\subsection{Proof of Theorem~\ref{thm:KiCQLF}}
\label{sec:thm_main}

\ref{item:4}) Since $\hat\Z$ is SLASF, then it is $\gamma$-SLASF for
some $0< \gamma \le 1$. Then, Theorem~\ref{thm:prevsic} shows that
$\hat\Z^\ell$ is $\gamma$-CEAS for $\ell=1,\ldots,n$, irrespective of
which $\hat v_1^\ell$ and $\hat F_i^\ell$ $\gamma$-compatible with
$\hat\Z^\ell$ are taken at each iteration $\ell$. Since $\hat\Z$ is
controllable, then $(\hat A_i^1 = \hat A_i,\hat B_i^1 = \hat B_i)$ is
controllable and $(\hat A_i^\ell,\hat B_i^\ell)$ is controllable for
$\ell=1,\ldots,n$ (see, for example, Proposition 1.2 of
\cite{wonham_book85}). 

Since $(\hat A_i,\hat B_i)$ is controllable and
$\hat B_i \in \C^{n\times 1}$, then $\hat B_i \neq 0$ and $\hat B_i$
has full column rank. By the continuity of (\ref{eq:49}), given
$\epsilon>0$ we can find $\epsilon_1>0$ such that (\ref{eq:2}) holds
whenever
\begin{gather}
  \label{eq:25}
  \|\hat A_i - A_i \| < \epsilon_1,\quad
  \|\hat B_i - B_i \| < \epsilon_1,\\
  \label{eq:28}
  \|\hat K_i - K_i \| < \epsilon_1.  
\end{gather}
By (\ref{eq:9}), given $\epsilon_1>0$ we can find
$0<\epsilon_2<\epsilon_1$ such that (\ref{eq:28}) holds whenever
\begin{equation}
  \label{eq:27}
  \| \hat F_i^\ell - F_i^\ell \| < \epsilon_2,\quad
  \| \hat U_\ell - U_\ell \| < \epsilon_2,\quad
  \text{for }\ell=1,\ldots,n.
\end{equation}

For $\ell=n$, $\hat A_i^n$ and $\hat B_i^n$ are scalars and $\hat
B_i^n \neq 0$ because $(\hat A_i^n,\hat B_i^n)$ is
controllable. Consider the following condition
\begin{equation}
  \label{eq:25a}
  \|\hat A_i^n - A_i^n \| < \epsilon_3,\quad
  \|\hat B_i^n - B_i^n \| < \epsilon_3.
\end{equation}
Note that if we select $0 < \epsilon_3 < \epsilon_2$ small enough then
(\ref{eq:25a}) will imply that $B_i^n \neq 0$. Taking such small
$\epsilon_3$, for each set $\{(A_i^n,B_i^n):i\in\Ind\}$ satisfying
(\ref{eq:25a}) Procedure~$\cea$ in Algorithm~\ref{alg:main} does not
employ $\epsilon_d^n$, and returns $F_i^n$ such that $|A_i^{n,\CL}|
\le 1-\epsilon_c^n$, which is possible for every $0<\epsilon_c^n\le
1$. Note that if we choose $\epsilon_3>0$ small enough, $F_i^n$ will,
in addition, satisfy $|\hat A_i^n + \hat B_i^n F_i^n|\le
1-\epsilon_c^n/2$. Therefore, we may take $\hat F_i^n = F_i^n$.

If $n=1$, \ref{item:4}) is established taking $\epsilon_c^\star = 1$,
arbitrary $\epsilon_d^\star > 0$, and $\delta=\epsilon_3$,
since in this case $K_i = F_i^1$ and by the above argument we may take
$\hat K_i = K_i$.

We proceed for $n>1$. We next establish the existence of
$\epsilon_c^\star$ and $\epsilon_d^\star$. For $\ell=1,\ldots, n-1$,
i.e. for $n_r > 1$, consider the sets $\hat\Tb_\gamma^\ell$ as
defined in (\ref{eq:setT}). 
Note that the set $\hat\Tb_\gamma^\ell$ depends on the matrices $\hat
A_i^\ell$ and $\hat B_i^\ell$, for $i\in\Ind$. Since the latter
matrices depend on the specific vectors and matrices computed at
previous iterations of Algorithm~\ref{alg:main}, then
$\hat\Tb_\gamma^\ell$ also depends on such quantities. 
For
$\ell=1,\ldots,n-1$ and $k=1,\ldots,\ell$, consider the expression
\begin{equation}
  \label{eq:53}
  \eta_\gamma^{k,\ell} \dfn \inf_{\hat v_1^k\in\hat\Tb_\gamma^k}
  \inf_{\hat v_1^{k+1}\in\hat\Tb_\gamma^{k+1}} \cdots \left[\inf_{\hat v_1^\ell\in\hat\Tb_\gamma^{\ell}}
  \min_{i\in\Ind} \dist(\hat v_1^\ell,\img \hat B_i^\ell)\right].
\end{equation}
Note that the expression between square brackets in (\ref{eq:53})
coincides with $\epsilon_d^{\ell,\star}$ as given in
(\ref{eq:48}). The expressions (\ref{eq:53}) can be interpreted as the
infimum of the minimum distance between all the possible
$\gamma$-compatible vectors that could be obtained at iteration $\ell$
and all the possible $\img \hat B_i^\ell$ that can be obtained at
iteration $\ell$, having the data corresponding to iteration
$k\le\ell$ and over all possible outcomes at iterations
$k,k+1,\ldots,\ell$. We make the following claim, to be proved later.
\begin{claim}
  \label{clm:quilombo}
  For $\ell=1,\ldots,n-1$ and $k=1,\ldots,\ell$, expression (\ref{eq:53})
  \begin{enumerate}[a)]
  \item may depend on $\hat A_i^k$, $\hat B_i^k$ and $\gamma$ but does not
    depend on $\hat F_i^r$ or $\hat U_{r+1}$ for $k\le r \le \ell-1$.\label{item:11}
  \item is positive.\label{item:10}
  \end{enumerate}
\end{claim}
By Claim~\ref{clm:quilombo}, knowing $\hat A_i^1 = \hat A_i$, $\hat
B_i^1 = \hat B_i$ and $\gamma$, we can define positive
\begin{align}
  \label{eq:64}
  \epsilon_c^\star &\dfn \gamma,\\
  \label{eq:63}
  \epsilon_d^\star &\dfn \min_{\ell=1,\ldots,n-1} \eta_\gamma^{1,\ell}.
\end{align}
Note that $\epsilon_d^{\star} \le \epsilon_d^{\ell,\star}$ for
$\ell=1,\ldots,n-1$, with $\epsilon_d^{\ell,\star}$ as in (\ref{eq:48}).


Applying Theorem~\ref{thm:num1} repeatedly
($n-1$ times), it follows that each $0 < \epsilon_c^\ell <
\epsilon_c^\star$ and $0 < \epsilon_d^\ell < \epsilon_d^\star$ ensure
that given $\delta_n \dfn \epsilon_3>0$, we can find a corresponding
$0 < \delta_\ell < \epsilon_3$, for $\ell=n-1,\ldots,1$ so that for
each $A_i^\ell,B_i^\ell$ satisfying
\begin{equation}
  \label{eq:31}
  \|\hat A_i^\ell - A_i^\ell \| < \delta_\ell,\quad
  \|\hat B_i^\ell - B_i^\ell \| < \delta_\ell,
\end{equation}
there exist\footnote{We add a superscript $\ell$ to the set $\Sb$ to
  denote that this set is different at each iteration of
  Algorithm~\ref{alg:main}.} $\hat v_1^\ell \in
\hat\Sb^\ell(\epsilon_c^\ell/2, \epsilon_d^\ell/2)$ and $\hat
F_i^\ell$ compatible with $\hat\Z^\ell$, and a unitary matrix
(\ref{eq:74}) that cause, for $\ell=1,\ldots,n-1$,
\begin{gather}
  \label{eq:32}
  \| \hat F_i^\ell - F_i^\ell \| < \delta_{\ell+1},\quad
  \| \hat v_1^\ell - v_1^\ell \| < \delta_{\ell+1},\\
  \label{eq:34}
  \| \hat U_{\ell+1} - U_{\ell+1} \| < \delta_{\ell+1},\quad
  \| \hat A_i^{\ell+1} - A_i^{\ell+1} \| < \delta_{\ell+1},\quad
  \| \hat B_i^{\ell+1} - B_i^{\ell+1} \| < \delta_{\ell+1}.
\end{gather}
By definition, $\hat U_1 = U_1 = \id$. The
latter fact jointly with (\ref{eq:32})--(\ref{eq:34}) establish
(\ref{eq:27}). Take $\delta = \delta_1$. Since $\hat v_1^\ell \in
\hat\Sb^\ell(\epsilon_c^\ell/2, \epsilon_d^\ell/2)$ and $\hat v_1^\ell$,
$\hat F_i^\ell$ are compatible with $\hat\Z^\ell$ then
Lemma~\ref{lem:comeig} shows that $(\hat A_i^\ell + \hat B_i^\ell \hat
F_i^\ell) \hat v_1^\ell = \lambda_i^\ell \hat v_1^\ell$ with
$|\lambda_i^\ell| \le 1-\epsilon_c^\ell/2$, for
$\ell=1,\ldots,n-1$. Also, $\lambda_i^\ell$ is an eigenvalue of $\hat
A_i^\CL$ (recall Remark~\ref{rem:ut}), and so is $\hat A_i^n + \hat
B_i^n \hat F_i^n$, which satisfies $|\hat A_i^n + \hat B_i^n \hat
F_i^n|\le 1-\epsilon_c^n/2$. Therefore, $\srad(\hat A_i^\CL)\le
1-\epsilon_c/2$ and the proof of part~\ref{item:4}) is concluded if
Claim~\ref{clm:quilombo} is true.

To establish Claim~\ref{clm:quilombo}, consider the sets
$\hat\Tb_\gamma^\ell$ defined in (\ref{eq:72}), for
$\ell=1,\ldots,n-1$. Note that $\hat\Tb_\gamma^{\ell}$ depends on
$\hat A_i^\ell, \hat B_i^\ell$ which, by (\ref{eq:33}) and
(\ref{eq:62})--(\ref{eq:6}), may depend on $\hat A_i^{\ell-1}, \hat
B_i^{\ell-1}, \hat F_i^{\ell-1}$ and $\hat U_{\ell}$. We next
establish that $\hat\Tb_\gamma^{\ell}$ does not depend on $\hat
F_i^{\ell-1}$. From (\ref{eq:33}), (\ref{eq:100}) and (\ref{eq:6}), we
have $\hat A_i^\ell = \hat U_\ell^* \hat A_i^{\ell-1} \hat U_\ell +
\hat B_i^{\ell} \hat F_i^{\ell-1} \hat U_\ell$. Consider $\hat
P_i^\ell$ as in (\ref{eq:P}), and note that $\hat P_i^\ell \hat
B_i^\ell = 0$. Thus, it follows that $\hat P_i^\ell (\lambda \id -
\hat A_i^\ell) = \hat P_i^\ell (\lambda \id - \hat U_\ell^* \hat
A_i^{\ell-1} \hat U_\ell)$, which does not depend on $\hat
F_i^{\ell-1}$. From the equivalent definition of
$\hat\Tb_\gamma^{\ell}$ in (\ref{eq:Talt})--(\ref{eq:P}), it follows
that $\hat\Tb_\gamma^{\ell}$ does not depend on $\hat
F_i^{\ell-1}$. Following similar considerations, we can show that
$\hat\Tb_\gamma^{\ell}$ does not depend on $\hat F_i^k$ for
$k=1,\ldots,\ell-1$.

Next, consider $\eta_\gamma^{\ell,\ell}$ from (\ref{eq:53}), rewritten
to explicitly show the dependence on matrices:
\begin{equation}
  \label{eq:75}
  \eta_\gamma^{\ell,\ell} =\inf_{\hat v^\ell\in\hat\Tb_\gamma^{\ell}(\hat A_i^{\ell-1},\hat
      B_i^{\ell-1},\hat U_\ell)}
  \min_{i\in\Ind} \dist(\hat v^\ell,\img \hat U_\ell^* \hat B_i^{\ell-1}),
\end{equation}
We next show that for fixed $\hat A_i^{\ell-1}$ and $\hat
B_i^{\ell-1}$, (\ref{eq:75}) depends only on the orthogonal complement
of $\hat U_\ell$ and not on $\hat U_\ell$ itself. Consider fixed $\hat
A_i^{\ell-1},\hat B_i^{\ell-1},\hat U_\ell$, and let $\hat v_1^{\ell-1}$
be such that $[\hat v_1^{\ell-1}|\hat U_\ell]$ is unitary. Let $\tilde
U_\ell$ also make $[\hat v_1^{\ell-1}|\tilde U_\ell]$ unitary. Write
$\hat B_i^{\ell-1}$ as
\begin{equation}
  \label{eq:76}
  \hat B_i^{\ell-1} = \hat v_1^{\ell-1} \alpha + \hat U_\ell \hat\beta_i =
  \hat v_1^{\ell-1} \alpha + \tilde U_\ell \tilde\beta_i
\end{equation}
Using (\ref{eq:76}) we may write
\begin{equation}
  \label{eq:77}
  \dist(\hat v^\ell,\img \hat U_\ell^* \hat B_i^{\ell-1}) =
  \dist(\hat v^\ell,\img \hat\beta_i) = \dist(\hat v^\ell,\img \hat U_\ell^* \tilde
  U_\ell \tilde\beta_i)
\end{equation}
Put $\hat v^\ell = \hat U_\ell^* \tilde U_\ell \tilde v^\ell$ and, since
$\hat U_\ell^* \tilde U_\ell$ is unitary, we have
\begin{equation}
  \label{eq:79}
  \dist(\hat v^\ell,\img \hat\beta_i) = \dist(\tilde v^\ell,\img \tilde\beta_i),\quad
  \text{and}\quad
  \min_{i\in\Ind}\dist(\hat v^\ell,\img \hat\beta_i) = \min_{i\in\Ind}\dist(\tilde v^\ell,\img \tilde\beta_i).
\end{equation}
Let $\tilde\Tb_\gamma^\ell = \hat\Tb_\gamma^\ell(\hat A_i^{\ell-1},\hat
B_i^{\ell-1},\tilde U_\ell)$ and note that $\hat v^\ell
\in \hat\Tb_\gamma^\ell$ if and only if $\tilde v^\ell \in
\tilde\Tb_\gamma^\ell$. Therefore,
\begin{equation}
  \inf_{\hat v^\ell\in\hat\Tb_\gamma^{\ell}} \min_{i\in\Ind}\dist(\hat v^\ell,\img \hat\beta_i) = 
  \inf_{\tilde v^\ell\in\tilde\Tb_\gamma^{\ell}} \min_{i\in\Ind}\dist(\tilde v^\ell,\img \tilde\beta_i)
\end{equation}
Recalling that $\hat\beta_i = \hat U_\ell^* \hat B_i^{\ell-1}$ and
$\tilde\beta_i = \tilde U_\ell^* \hat B_i^{\ell-1}$, we establish that
(\ref{eq:75}) depends on $\hat U_{\ell}$ only through its
orthogonal complement $\hat v_1^{\ell-1}$. Following similar
considerations, we can establish that $\eta_\gamma^{k,\ell}$ depends
on $\hat U_r$ only through its orthogonal complement $\hat v^{r-1}$,
for $r=k,\ldots,\ell$.

To prove part~\ref{item:10}) of Claim~\ref{clm:quilombo}, consider
(\ref{eq:75}) again. Since $\eta_\gamma^{\ell,\ell} =
\epsilon_d^{\ell,\star}$ as in (\ref{eq:48}), we already know that
$\eta_\gamma^{\ell,\ell}>0$ for every possible $\hat
v_1^{\ell-1}$. Next, for fixed $\hat A_i^{\ell-1}$, $\hat
B_i^{\ell-1}$, we may write
\begin{align}
  \label{eq:80}
  \eta_\gamma^{\ell,\ell} &= -g(\hat v_1^{\ell-1}) = \inf_{\hat
    v_1^\ell\in\Rb(\hat U_\ell)} -f(\hat U_\ell, \hat v_1^\ell) =
  -\sup_{\hat v_1^\ell\in\Rb(\hat U_\ell)} f(\hat U_\ell, \hat v_1^\ell),\\
  \label{eq:84}
  \eta_\gamma^{\ell-1,\ell} &= \inf_{\hat v_1^{\ell-1} \in
    \hat\Tb_\gamma^{\ell-1}} -g(\hat v_1^{\ell-1}),
\end{align}
where, in (\ref{eq:80}), $\hat U_\ell$ must be taken so that $[\hat
v_1^{\ell-1}|\hat U_\ell]$ is unitary, and we have defined
\begin{align}
  \label{eq:81}
  f(\hat U_\ell, \hat v_1^\ell) &= -\min_{i\in\Ind} \dist(\hat v_1^\ell,\img \hat U_\ell^* \hat B_i^{\ell-1}),\\
  \label{eq:83}
  \Rb(\hat U_\ell) &= \hat\Tb_\gamma^{\ell}(\hat A_i^{\ell-1},\hat B_i^{\ell-1},\hat U_\ell).
\end{align}
Note that $f$ is continuous at every $(\hat U_\ell,\hat v^\ell)$ for
which $\hat U_\ell^* \hat B_i^{\ell-1} \neq 0$ for $i\in\Ind$. 

Let $\setU$ denote the following set
\begin{equation}
  \label{eq:94}
  \setU := \{ \hat U_\ell \in \C^{n-\ell,n-\ell-1} :  
  [\hat v_1^{\ell-1}|\hat U_\ell]\text{ is unitary, } \hat v_1^{\ell-1} \in \hat\Tb_\gamma^{\ell-1} \}
\end{equation}
and regard $\Rb$ as the set-valued map $\Rb : \setU \leadsto \Sb_1$
that maps each $\hat U_\ell\in\setU$ to the set $\Rb(\hat
U_\ell)\subset\Sb_1$. Making these considerations, the function
$g:\hat\Tb_\gamma^{\ell-1} \to \R$ defined above can be regarded as a
marginal function. Application of Theorem 1.4.16 of
\cite{aubfra_book90} to $g$ yields that $g$ is upper semicontinuous if
$f$ and $\Rb$ are upper semicontinuous, and $\Rb$ has compact
values. Since $f$ is continuous, it is upper semicontinuous. $\Rb$
takes compact values because $\hat\Tb_\gamma^\ell$ is compact.
\begin{claim}
  \label{clm:usc}
  The set-valued map $\Rb : \setU \leadsto \Sb_1$ is upper semicontinuous.
\end{claim}
\begin{proof}
  Since $\Rb$ has compact values, i.e. the sets $\Rb(\hat U_\ell)$ are
  compact for every $\hat U_\ell \in \setU$, we need only prove that
  the graph of $\Rb$ is closed. For, take a sequence $\{u_k \in
  \setU\}_{k=0}^\infty$ so that $\lim_{k\to\infty} u_k = \hat U_\ell$
  and suppose that $v_k \in \Rb(u_k)$ so that $\lim_{k\to\infty} v_k =
  v \in \Sb_1$. We have to establish that $v\in\Rb(\hat
  U_\ell)$. Recalling (\ref{eq:Talt})--(\ref{eq:P}), we have that $v_k
  \in \Rb(u_k)$ implies
  \begin{equation}
    \label{eq:87}
    \hat P_{i,k}^\ell (\lambda_{i,k} \id - u_k^* \hat A_i^{\ell-1} u_k) v_k = 0,
  \end{equation}
  with $|\lambda_{i,k}|\le 1-\gamma$ and
  \begin{equation}
    \label{eq:88}
    \hat P_{i,k}^\ell = \left[\id - u_k^* \hat B_i^{\ell-1}
      \left((u_k^* \hat B_i^{\ell-1})^* u_k^* \hat
        B_i^{\ell-1}\right)^{-1}(u_k^* \hat B_i^{\ell-1})^*\right].
  \end{equation}
  for all $i\in\Ind$ and $k\ge 0$. We may rewrite (\ref{eq:87}) as
  \begin{equation}
    \label{eq:89}
    \hat P_{i,k}^\ell v_k \lambda_{i,k} = \hat P_{i,k}^\ell u_k^* \hat A_i^{\ell-1} u_k v_k.
  \end{equation}
  Since $u_k$ and $v_k$ are convergent, the limit of the right-hand
  side of (\ref{eq:89}) exists and we have
  \begin{equation}
    \label{eq:90}
    \lim_{k\to\infty} \hat P_{i,k}^\ell v_k \lambda_{i,k} = 
    \hat P_i^\ell \hat U_\ell^* \hat A_i^{\ell-1} \hat U_\ell v = 
    \lim_{k\to\infty} \hat P_{i}^\ell v \lambda_{i,k}, 
  \end{equation}
  where the last equality above follows because $\lambda_{i,k}$ is
  bounded. If $\hat P_i^\ell v \neq 0$, then $\lim_{k\to\infty}
  \lambda_{i,k} = \lambda_i$, with $|\lambda_i|\le 1-\gamma$. If $\hat
  P_i^\ell v = 0$, then (\ref{eq:90}) equals zero. In either case, we
  have $v\in\Rb(\hat U_\ell) = \hat\Tb_\gamma^\ell(\hat
  A_i^{\ell-1},\hat B_i^{\ell-1},\hat U_\ell)$.
\end{proof}
The continuity of $f$ and Claim~\ref{clm:usc} imply that $g$ is upper
semicontinuous. Then, $-g$ is lower semicontinuous and since
$\hat\Tb_\gamma^{\ell-1}$ is compact, then $-g$ will attain a minimum
within $\hat\Tb_\gamma^{\ell-1}$. Therefore,
$\eta_\gamma^{\ell-1,\ell} > 0$. Following similar considerations, we
can establish that $\eta_\gamma^{k,\ell} > 0$ for $k=1,\ldots,\ell$.

\ref{item:7}) By Definition~\ref{def:SLASF}, a SLASF $\hat\Z$ with
compatible $\hat K_i$ means that $\hat A_i^\CL$ will generate a
solvable Lie algebra and satisfy $\srad(\hat A_i^\CL)<1$. In addition,
part~\ref{item:4}) ensures that $\srad(\hat A_i^\CL) \le
1-\epsilon_c/2$ irrespective of which $\epsilon$ or $\delta$ are
selected. By Lemma~\ref{lem:stab}, the closed-loop DTSS with
subsystem matrices $\hat A_i^\CL$ admits a CQLF. Part~\ref{item:7})
then follows from robustness of the CQLF, selecting $\epsilon > 0$
small enough.

\bibliographystyle{IEEEtran}

\begin{thebibliography}{10}
\providecommand{\url}[1]{#1}
\csname url@rmstyle\endcsname
\providecommand{\newblock}{\relax}
\providecommand{\bibinfo}[2]{#2}
\providecommand\BIBentrySTDinterwordspacing{\spaceskip=0pt\relax}
\providecommand\BIBentryALTinterwordstretchfactor{4}
\providecommand\BIBentryALTinterwordspacing{\spaceskip=\fontdimen2\font plus
\BIBentryALTinterwordstretchfactor\fontdimen3\font minus
  \fontdimen4\font\relax}
\providecommand\BIBforeignlanguage[2]{{%
\expandafter\ifx\csname l@#1\endcsname\relax
\typeout{** WARNING: IEEEtran.bst: No hyphenation pattern has been}%
\typeout{** loaded for the language `#1'. Using the pattern for}%
\typeout{** the default language instead.}%
\else
\language=\csname l@#1\endcsname
\fi
#2}}

\bibitem{liberzon99:_basic}
D.~Liberzon and A.~S. Morse, ``Basic problems in stability and design of
  switched systems,'' \emph{IEEE Control Systems Magazine}, vol.~19, no.~5, pp.
  59--70, 1999.

\bibitem{molchanov89:_criter}
A.~P. Molchanov and Y.~S. Pyatnitskiy, ``Criteria of asymptotic stability of
  differential and difference inclusions encountered in control theory,''
  \emph{Systems and Control Letters}, vol.~13, no.~1, pp. 59--64, 1989.

\bibitem{shorten07:_stabil}
R.~Shorten, F.~Wirth, O.~Mason, K.~Wulff, and C.~King, ``Stability criteria for
  switched and hybrid systems,'' \emph{SIAM Review}, vol.~49, no.~4, pp.
  545--592, 2007.

\bibitem{daafouz02:_stabil}
J.~Daafouz, P.~Riedinger, and C.~Iung, ``Stability analysis and control
  synthesis for switched systems: A switched {Lyapunov} function approach,''
  \emph{IEEE Transactions on Automatic Control}, vol.~47, no.~11, pp.
  1883--1887, 2002.

\bibitem{sala05:_comput}
A.~Sala, ``Computer control under time-varying sampling period: an {LMI}
  gridding approach,'' \emph{Automatica}, vol.~41, no.~12, pp. 2077--2082,
  2005.

\bibitem{wulff09:_contr_dsgn}
K.~Wulff, F.~Wirth, and R.~Shorten, ``A control design method for a class of
  switched linear systems,'' \emph{Automatica}, vol.~45, no.~11, pp.
  2592--2596, 2009.

\bibitem{lin09:_stabil}
H.~Lin and P.~J. Antsaklis, ``Stability and stabilisability of switched linear
  systems: a survey of recent results,'' \emph{IEEE Trans.\ on Automatic
  Control}, vol.~54, no.~2, pp. 308--322, February 2009.

\bibitem{kokarc_aut01}
P.~Kokotovi\'{c} and M.~Arcak, ``Constructive nonlinear control: a historical
  perspective,'' \emph{Automatica}, vol.~37, pp. 637--662, 2001.

\bibitem{HBF2009:_CDC}
H.~Haimovich, J.~H. Braslavsky, and F.~Felicioni, ``On feedback stabilisation
  of switched discrete-time systems via {L}ie-algebraic techniques,'' in
  \emph{Proceedings of the Joint 48th IEEE Conference on Decision and Control
  and 28th Chinese Control Conference}, Shanghai, China, December 2009, pp.
  1118--1123, dOI: 10.1109/CDC.2009.5399527. Also submitted to the \emph{IEEE
  Trans.\ on Automatic Control}.

\bibitem{liberzon03:_switc}
D.~Liberzon, \emph{Switching in systems and control}.\hskip 1em plus 0.5em
  minus 0.4em\relax Birkh\"{a}user, 2003.

\bibitem{theys_phd05}
J.~Theys, ``Joint spectral radius: theory and approximations,'' Ph.D.
  dissertation, Center for Systems Engineering and Applied Mechanics,
  Universit\'e catholique de Louvain, 2005.

\bibitem{erdwil_book06}
K.~Erdmann and M.~Wildon, \emph{Introduction to {L}ie algebras}.\hskip 1em plus
  0.5em minus 0.4em\relax Springer-Verlag London, 2006.

\bibitem{mormor_cdc97}
Y.~Mori, T.~Mori, and Y.~Kuroe, ``A solution to the common {L}yapunov function
  problem for continuous-time systems,'' in \emph{Proc. 36th IEEE Conf. on
  Decision and Control, San Diego, CA, USA}, 1997, pp. 3530--3531.

\bibitem{wonham_book85}
W.~M. Wonham, \emph{Linear multivariable control: a geometric approach},
  3rd~ed.\hskip 1em plus 0.5em minus 0.4em\relax New York: Springer-Verlag,
  1985.

\bibitem{aubfra_book90}
J.-P. Aubin and H.~Frankowska, \emph{Set-valued analysis}.\hskip 1em plus 0.5em
  minus 0.4em\relax Birkhäuser Boston, 1990.

\end{thebibliography}

\end{document}